\documentclass[12pt]{article}
\usepackage{graphicx} 
\usepackage{epstopdf}
\usepackage{amsthm,amsmath,amssymb,epsfig}
\usepackage{algorithm, algorithmic}
\usepackage{booktabs}
\usepackage[authoryear,round]{natbib} 

\newtheorem{theorem}{Theorem}
\newtheorem{lemma}{Lemma}

\newtheorem{corollary}{Corollary}

\def\btheta{\mbox{\boldmath $\theta$}}
\def\bepsilon{\mbox{\boldmath $\epsilon$}}

\def\bSigma{\mbox{\boldmath $\Sigma$}}

\def\bw{\mbox{\boldmath $\beta$}}

\def\bR{\mathbb{R}}
\def\mE{\mathbb{E}}
\def\bP{\mathbb{P}}
\def\mQ{\mathbb{Q}}

\def\bnu{\mbox{\boldmath $\nu$}}
\def\ba{{\bf a}}
\def\bb{{\bf b}}

\def\bQ{{\bf Q}}

\def\bw{{\bf w}}
\def\bs{{\bf s}}

\def\bu{{\bf u}}
\def\bv{{\bf v}}

\def\bI{{\bf I}}

\def\bx{{\bf x}}
\def\by{{\bf y}}
\def\bz{{\bf z}}

\def\nn{\nonumber}

\def\v2{\vspace{0.2in}}

\numberwithin{equation}{section}

\newcommand{\std}[1]{\scriptsize\rm$\pm$#1}

\begin{document}

\title{A Closed-Form Framework for Schrödinger Bridges Between Arbitrary Densities}

\author{Hanwen Huang  \\\\
     {\it Department of Biostatistics, Data Science and Epidemiology}\\
    {\it Medical College of Georgia}\\
    {\it Augusta University, Augusta, GA 30912}\\
HHUANG1@augusta.edu}

\date{}

\maketitle

\begin{abstract}
Score-based generative models have recently attracted significant attention for their ability to generate high-fidelity data by learning maps from simple Gaussian priors to complex data distributions. A natural generalization of this idea to transformations between arbitrary probability distributions leads to the Schrödinger Bridge (SB) problem. However, SB solutions rarely admit closed-form expressions and are commonly obtained through iterative stochastic simulation procedures, which are computationally intensive and can be unstable. In this work, we introduce a unified closed-form framework for representing the stochastic dynamics of SB systems. Our formulation subsumes previously known analytical solutions—including the Schrödinger–Föllmer process and the Gaussian SB—as specific instances. Notably, the classical Gaussian SB solution, previously derived using substantially more sophisticated tools such as Riemannian geometry and generator theory, follows directly from our formulation as an immediate corollary.  Leveraging this framework, we develop a simulation-free algorithm that infers SB dynamics directly from samples of the source and target distributions. We demonstrate the versatility of our approach in two settings: (i) modeling developmental trajectories in single-cell genomics and (ii) solving image restoration tasks such as inpainting and deblurring. This work opens a new direction for efficient and scalable nonlinear diffusion modeling across scientific and machine learning applications.
\end{abstract}


\section{Introduction}
Diffusion-based generative models have seen rapid development and broad application across machine learning domains, including computer vision \citep{sohldickstein2015deep}, image analysis \citep{wang2021geometry,pu2015deep}, natural language processing \citep{inproceedings,yao2019DGM4NLP}, and drug discovery \citep{16272714120230328}. These models construct a stochastic process that transforms a simple reference distribution, such as a multivariate Gaussian, into a complex data distribution. A central class within this family is score-based generative models (SGMs) \citep{sohldickstein2015deep,ho2020denoising,song2021scorebased,song2021maximum}, which currently achieve state-of-the-art sample quality across multiple domains. SGMs couple two stochastic processes: a forward diffusion process that gradually corrupts data into noise, and a reverse denoising process that reconstructs data from noise. Both can be described as stochastic differential equations whose endpoints match the data and prior distributions. 

While SGMs excel at data synthesis, many scientific applications—such as single-cell analysis—require learning the stochastic dynamics that interpolate between two observed distributions, rather than generating samples from noise. In these settings, the goal is to infer the most likely probabilistic paths connecting initial and terminal state distributions. This has led to growing interest in the Schrödinger Bridge (SB) framework \citep{bernton2019schrodinger,chen2023likelihood,debortoli2023diffusion,shi2023diffusion,chen2023schrodinger,liu2023schrodinger,pmlr-v206-stromme23a}, which generalizes diffusion-based models. SB seeks a stochastic process that evolves between two marginals while staying close to a reference process, typically Brownian motion.

In general, SB problems lack closed-form solutions except in special cases (e.g., both marginals Gaussian \citep{bunne2023schrodingerbridgegaussianmeasures}). As a result, most methods rely on iterative numerical approximations \citep{debortoli2023diffusion,chen2023likelihood}, which can be unstable and scale poorly in high dimensions \citep{shi2023diffusion}. To overcome these computational limitations, we propose a closed-form framework for learning the stochastic dynamics between two arbitrary distributions. We derive, for the first time, an explicit analytical relation between the dynamic formulation of the SB problem and the solution to its corresponding static problem. Specifically, we show that the solutions of the SB system admit analytic expressions that depend solely on integrations with respect to the endpoint joint distributions. Remarkably, this result yields a unifying perspective that encompasses all previously known closed-form SB solutions, including the Schrödinger–Föllmer process—which transports a degenerate Dirac delta distribution at time zero to a target distribution at time one—and the Gaussian Schrödinger Bridge, which transports one multivariate Gaussian to another, as special cases.

Establishing this closed-form link between the dynamic SB formulation and static endpoint densities enables an efficient and tractable computation of interpolation dynamics, implemented as Algorithm 1, without resorting to iterative proportional fitting or neural network–based regression. For paired data, Algorithm 1 can be directly applied to compute the SB drift. For unpaired data, paired samples can first be generated using an optimal transport algorithm and then used within Algorithm 1 to recover the stochastic dynamics. Furthermore, our formulation admits an equivalent variational (minimization) representation, from which we derive a simulation-free algorithm, implemented as Algorithm 2, for inferring stochastic dynamics directly from samples drawn from arbitrary source and target distributions. This method provides a training-free framework that avoids the computational complexity and instability often associated with iterative, training-based approaches. Moreover, it generalizes and unifies several existing simulation-free formulations—including those proposed in \cite{liu2023schrodinger,tong2024}—as special cases.

Beyond SB-based methods, several other approaches have recently been proposed for learning nonlinear transformations between arbitrary distributions. A prominent class among them is the flow-based models, or continuous normalizing flows (CNFs), which assume a deterministic continuous-time generative process governed by an ordinary differential equation (ODE) that transports the source density to the target density. Similar to score-based generative models, CNFs admit closed-form solutions only when the source distribution is Gaussian. For arbitrary source distributions, their applicability has been limited by simulation-based training objectives that require expensive numerical integration of the ODE during training. Recent advances have introduced simulation-free training objectives that improve the efficiency of CNFs and make them competitive with SGMs. However, these methods still require training neural networks and, more importantly, do not extend to stochastic dynamics, which are essential both for generative modeling and for recovering the underlying dynamics of complex systems.

Our specific contributions are summarized as follows:
\begin{itemize}
\item We develop a closed-form framework for studying Schrödinger Bridge (SB) systems that connect arbitrary probability distributions. This framework unifies and generalizes all previously known closed-form SB solutions, including the Schrödinger–Föllmer process and the Gaussian Schrödinger Bridge, as special cases. In particular, the well-known closed-form solution to the Gaussian Schrödinger Bridge problem—previously derived using considerably more sophisticated tools such as Riemannian geometry and generator theory in \cite{bunne2023schrodingerbridgegaussianmeasures}—emerges as an immediate corollary of our framework (see Corollary 1). To the best of our knowledge, this constitutes the most general closed-form result to date in the line of research on analytical solutions to the SB problem.
\item Building on this framework, we introduce a training-free algorithm (Algorithm 1) for learning stochastic dynamics from both paired and unpaired datasets. This one-step method eliminates the need for neural network training and the associated challenges of architecture selection, making it computationally efficient and easy to implement. The algorithm proves particularly effective for image restoration and modeling the temporal evolution of single-cell genomic data.
\item Through an equivalent variational formulation, we further derive a simulation-free algorithm (Algorithm 2) for inferring stochastic dynamics directly from samples drawn from arbitrary source and target distributions. Unlike prior simulation-free approaches—such as those in \cite{liu2023schrodinger,tong2024}—which restrict the reference process to standard Brownian motion, our formulation accommodates arbitrary Itô processes as references. As demonstrated in Section~\ref{numerical}, careful selection of the reference process can lead to substantial performance gains.
\end{itemize}

The remainder of the paper is organized as follows. Section \ref{method} introduces the Schrödinger Bridge diffusion process and presents a tractable formulation that can be used to derive explicit closed-form solutions to the Schrödinger systems. Section \ref{setting} develops two algorithms: a training-free algorithm (Algorithm 1) in Section \ref{setting1}, derived directly from our main theorem, and a simulation-free algorithm (Algorithm 2) in Section \ref{setting2}, based on an equivalent variational representation of the same result. Section \ref{numerical} presents numerical experiments evaluating the performance of our methods on both simulated low-dimensional data and real-world high-dimensional image datasets. Finally, Section~\ref{conclusion} provides concluding remarks, and  the proofs of the main theoretical results are included in the supplimentary material.

\section{A closed-form framework for solving Schr\"{o}dinger Bridge problem}\label{method}

In this section, we first provide background on the Schrödinger Bridge (SB) process. We then introduce a tractable formulation that explicitly connects the dynamic formulation of the SB problem to the solution of its corresponding static problem. Finally, we apply this formulation to the Gaussian SB case and rederive its well-known closed-form solutions—previously obtained through considerably more sophisticated frameworks such as Riemannian geometry and generator theory.

\subsection{Background on Schr\"{o}dinger Bridge}
We are interested in find nonlinear functions that transfor one distribution $\mu_0(\bx)$ to another distribution $\mu_1(\bx)$ for $\bx\in\bR^d$. Our algorithms are implemented through constructing a stochastic diffusion process $\{\bx_t\}_{t=0}^1$ indexed by a continuous time variable $t\in[0,1]$ such that $\bx_0\sim\mu_{0}$, the initial distribution, and $\bx_1\sim\mu_1$, the target distribution. Denote by $\Omega=C([0,1],\bR^d)$ the space consisting of all $\bR^d$-valued continuous functions on the time interval $[0,1]$ and $\bP$ the measure over $\Omega$ induced by the following SDE
\begin{eqnarray}\label{sde0}
d\bx_t=\bb(\bx_t,t) dt+\sigma(t) d\bw_t,&&\bx_0\sim\mu_0,
\end{eqnarray}
where $\bw_t$ is the standard $d$-dimensional Brownian motion, $\bb(\cdot,\cdot):\bR^d\times[0,1]\rightarrow\bR^d$ is a vector-valued function called the drift coefficient of $\bx_t$, and $\sigma(t):\bR\rightarrow\bR$ is a scalar function known as the diffusion coefficient of $\bx_t$. When $\bb(\bx,t)=0$ and $\sigma(t)=1$, $\bx_t$ is just the standard $d$-dimensional Brownian motion. The SDE (\ref{sde0}) has a unique strong solution as long as the drift coefficient function $\bb(\bx,t)$ is globally Lipschitz in both state $\bx$ and time $t$ \citep{10.5555/129416}. 

Denote by $\mQ_t$ the marginal probability law at time $t$ for the probability measure $\mQ$ on $\Omega$. We write ${\cal D}(\mu_{0},\mu_1)=\{\mQ:\mQ_0=\mu_{0},\mQ_1=\mu_1\}$ for the set of all path measures with given marginal distribution $\mu_{0}$ at the initial time and $\mu_1$ at the final time. Then the solution of the SB problem with respect to the reference measure $\bP$ can be formulated by the following optimization problem
\begin{eqnarray}\label{kld0}
\mQ^\star&=&\text{argmin}_{\mQ\in{\cal D}(\mu_{0},\mu_1)}D(\mQ\|\bP),
\end{eqnarray}
where $D(\mQ\|\bP)$ denotes the relative entropy between two probability measures on $\mQ$ and $\bP$ which is defined as
\begin{eqnarray}\label{kldef}
D(\mQ\|\bP)&=&\left\{\begin{array}{cc}\int\log(d\mQ/d\bP)d\mQ&if~\mQ\ll\bP\\
\infty&otherwise\end{array}\right.,
\end{eqnarray}
where $\mQ\ll\bP$ denotes that $\mQ$ is absolutely continuous w.r.t. $\bP$ and $d\mQ/d\bP$ represents the Radon-Nikodym derivative of $\mQ$ w.r.t. $\bP$.

\subsection{A closed-form fromulation to the Schr\"{o}dinger Bridge problem}
Denote $\Pi(\mu_0,\mu_1)$ the set of all couplings between $\mu_0$ and $\mu_1$. The static Schrodinger bridge consists in finding the joint distribution $\mu^\star(\bx_0,\bx_1)\in\Pi(\mu_0,\mu_1)$ which is closest to the reference prior subject to marginal constraints, that is
\begin{eqnarray}\label{static}
\mu^\star=\text{argmin}_{\mu(\bx_0,\bx_1)\in\Pi(\mu_0,\mu_1)}D(\mu\|\bP_{01}),
\end{eqnarray}
where $\bP_{01}(\bx_0,\bx_1)$ denotes the joint distribution at times $0$ and $1$ for the reference measure $\bP$ on $\Omega$ which is given by $\bP_{01}(\bx_0,\bx_1)=\mu_0(\bx_0)q(0,\bx_0,1,\bx_1)$, where with $q(t_1,\bx,t_2,\by)$ denotes the transition density of $\bx_{t_2}=\by$ at time $t_2$ given $\bx_{t_1}=\bx$ at time $t_1$ for stochastic process $\bx_t$ governed by the reference SDE (\ref{sde0}). For Brownian motion reference process, i.e. $\bb(\bx,t)=0$ and $\sigma(t)=\sigma$, (\ref{static}) becomes the entropically-regularized optimal transport problem defined as
\begin{eqnarray}\label{static0}
\mu^\star=\text{argmin}_{\mu(\bx_0,\bx_1)\in\Pi(\mu_0,\mu_1)}\left\{\int\frac{\|\bx_1-\bx_0\|^2}{2\sigma^2}d\mu(\bx_0,\bx_1)+D(\mu\|\mu_0\otimes\mu_1)\right\},
\end{eqnarray}
where $\mu_0\otimes\mu_1$ is the joint distribution over $\bx_0,\bx_1$ in which $\bx_0$ and $\bx_1$ are independent.

Denote the gradient of a smooth function  $f(\bx)$ by $\nabla f(\bx)$ and the partial derivative with respect to $\bx$ of $\psi(\bx,\by)$ for $(\bx, \by) \in \mathbb{R}^d\times\mathbb{R}^d$ by $\nabla_{\bx} \psi(\bx,\by)$. 
We state in the following theorem the solution $\mQ^\star$ to the SB problem (\ref{kld0}) whose reference measure $\bP$ is induced by the SDE (\ref{sde0}).
\begin{theorem}\label{thm}
Assume that $\sigma(t)\in C^1([0,1])$ and the components of $\bb(\bx,t)$ are bounded continuous and satisfy H\"older conditions with respect to $\bx$, i.e., there are real constants $C\ge 0,~\alpha\textgreater 0$ such that $|b_i(\bx,t)-b_i(\by,t)|\le C\|\bx-\by\|^\alpha$ for all $i=1,\cdots,d$ and $\bx,\by\in\bR^d$. Then the solution probability measures $\mQ^\star$ to the SB problem (\ref{kld0}) is induced by the following SDE with a modified drift:
\begin{eqnarray}\label{sde}
d\bx_t=[\bb(\bx_t,t)+\bu^\star(\bx_t,t)] dt+\sigma(t) d\bw_t,~~~\bx_0\sim\mu_0,
\end{eqnarray}
where the drift term $\bu^\star(\bx,t)$ is given by
\begin{eqnarray}\label{drift0}
\bu^\star(\bx,t)&=&\frac{\sigma(t)^2\int s_{\bx_1}(\bx,t)g_t(\bx,\bx_0,\bx_1)\mu^\star(d\bx_0,d\bx_1)}{\int g_t(\bx,\bx_0,\bx_1)\mu^\star(d\bx_0,d\bx_1)},
\end{eqnarray}
where $\mu^\star$ is the solution to the static problem (\ref{static}), $s_{\bx_1}(\bx,t)$ is the conditional score defined as 
\begin{eqnarray}\label{score}
\bs_{\bx_1}(t,\bx_t) = \nabla_{\bx_t}\log q(t,\bx_t,1,\bx_1),
\end{eqnarray}
and
\begin{eqnarray}\label{transition}
g_t(\bx,\bx_0,\bx_1)&=&\frac{q(0,\bx_0,t,\bx)q(t,\bx,1,\bx_1)}{q(0,\bx_0,1,\bx_1)}.
\end{eqnarray}
\end{theorem}
The proof is given in Appendix \ref{proof}. Theorem \ref{thm} shows that we can start from initial sample $\bx_0\sim\mu_0$ and update the values of $\{\bx_t : 0 < t \le 1\}$ according to the SDE (\ref{sde}) in continuous time, the value $\bx_1$ has the desired distributional property, that is, $\bx_1 \sim \mu_1$. 

For the solution to the SB problem (\ref{sde}), equation (\ref{drift0}) provides an explicit relation between the drift term and the solution to its corresponding static problem. This formulation is general, encompassing all previously known closed-form results in this area as special cases. For instance, if $\mu_{0}=\delta_{\ba}$, the Dirac delta distribution centered at $\ba\in\bR^d$, then $\mu^\star(d\bx_0,d\bx_1)=\delta_{\ba}(\bx_0)\mu_1(\bx_1)$, and (\ref{drift0}) simplifies to
\begin{eqnarray}\label{drifts}
\bu^\star(\bx,t)&=&\frac{\sigma(t)^2\int s_{\bx_1}(\bx,t)\frac{q(t,\bx,1,\bx_1)}{q(0,\ba,1,\bx_1)}\mu_1(d\bx_1)}{\int\frac{q(t,\bx,1,\bx_1)}{q(0,\ba,1,\bx_1)}\mu_1(d\bx_1)}.
\end{eqnarray}
This is a well-known result that has been studied in \citet{Pavon1989,DaiPra1991,Leonard2014,TzenR19,huang2024}. When the reference process $\bP$ is induced by (\ref{sde0}) with $\bb(\bx_t,t)=0$ and $\sigma(t)=\sigma$, i.e., when $\bP$ is a Wiener measure, the SB problem (\ref{kld0}) with initial marginal $\delta_{\ba}$ is referred to as the Schr\"{o}dinger–F\"{o}llmer process (SFP) \citep{Follmer}. The simplified expression in (\ref{drifts}) has been recently applied to generative modeling \citep{huang2021schrodingerfollmer,huang2024} and Bayesian inference \citep{vargas2022bayesian,zhang2022path}.

Moreover, the drift function (\ref{drifts}) can be equivalently characterized as the solution to the following stochastic control problem:
\begin{eqnarray}\label{optimal}
\bu^\star(\bx_t,t)&=&\text{argmin}_{\bu(\bx_t,t)}\mE\left[\frac{1}{2}\int|\bu(\bx_t,t)|^2dt\right],\\\nn
\text{s.t.}&&\left\{\begin{array}{c}d\bx_t=[\bb(\bx_t,t)+\bu(\bx_t,t)]dt+\sigma(t)d\bw_t,\\\bx_0=\delta_{\ba},~~~\bx_1\sim\mu_1.\end{array}\right.
\end{eqnarray}
An appealing feature of this control formulation is that it naturally yields efficient sampling schemes by transporting particles from any fixed point $\ba$ to samples drawn from the target distribution $\mu_1$ over the unit time interval \citep{zhang2022path}.

We can simplify (\ref{drift0}) by introducing the conditional distributions: $\pi(\bx_1|\bx)\equiv q(t,\bx,1,\bx_1)$, $\pi(\bx|\bx_0)\equiv q(0,\bx_0,t,\bx)$, and $\pi(\bx_1|\bx_0)\equiv q(0,\bx_0,1,\bx_1)$. Using the definition of $g_t(\bx,\bx_0, \bx_1)$ in (\ref{transition}), we compute the additional drift term as:
\begin{eqnarray}\nn
\bu^\star(\bx,t)&=&\frac{\sigma(t)^2\int\frac{\bs_{\bx_1}(t,\bx)\pi(\bx|\bx_0)\pi(\bx_1|\bx)}{\pi(\bx_1|\bx_0)}\mu^\star(d\bx_0,d\bx_1)}{\int \frac{\pi(\bx|\bx_0)\pi(\bx_1|\bx)}{\pi(\bx_1|\bx_0)}\mu^\star(d\bx_0,d\bx_1)}
\\\nn
&=&\frac{\sigma(t)^2\int\bs_{\bx_1}(t,\bx)\pi(\bx|\bx_0,\bx_1)\mu^\star(d\bx_0,d\bx_1)}{\int\pi(\bx|\bx_0,\bx_1)\mu^\star(d\bx_0,d\bx_1)}\\\nn
&=&\sigma(t)^2\int\bs_{\bx_1}(t,\bx)\pi(\bx_0,\bx_1|\bx)\mu^\star(d\bx_0,d\bx_1)\\\label{ualt}
&=&\sigma(t)^2\mE_{\pi(\bx_0,\bx_1|\bx)}\bs_{\bx_1}(t,\bx),
\end{eqnarray}
where the expectation is taken over the conditional joint distribution of $\bx_0$ and $\bx_1$ given $\bx_t = \bx$. 

\subsection{Exact Solution of the Schrödinger Bridge Problem for Multivariate Gaussian Measures}\label{refsde}
Note that in (\ref{drift0}), the transition kernel corresponds to the reference process defined by (\ref{sde0}). To simplify computation, one may select a reference SDE that admits a closed-form solution. In this paper, we consider a general class of reference measures $\bP$ that encompasses most processes commonly used in machine learning applications of Schrödinger bridges (SBs). Specifically, given an initial condition $\bx_0$, we define $\bP$ as the measure induced by the following linear stochastic differential equation (SDE):
\begin{eqnarray}\label{affine}
d\bx_t&=&{c(t)\bx_t+\alpha(t)}dt+\sigma(t)d\bw_t,
\end{eqnarray}
where $c(\cdot):[0,1]\rightarrow\bR$ and $\alpha(\cdot):[0,1]\rightarrow\bR^d$.
Let $\tau(t)=\exp\left(\int_0^t c(s)ds\right)$. Then, the solution to (\ref{sde0}) can be expressed as
\begin{eqnarray}\label{xt}
\bx_t=\tau(t)\left(\bx_0+\int_0^t\frac{\alpha(s)}{\tau(s)},ds+\int_0^t\frac{\sigma(s)}{\tau(s)}d\bw_s\right).
\end{eqnarray}
Using the independent increments of $\bx_t$ and Itô’s isometry, we obtain
\begin{eqnarray}\nn
\mE(\bx_t|\bx_0)=\tau(t)\left(\bx_0+\int_0^t\frac{\alpha(s)}{\tau(s)}ds\right)=\eta(t),
\end{eqnarray}
and, for any $t'\ge t$,
\begin{eqnarray}\nn
\mE[(\bx_t-\eta(t))(\bx_{t'}-\eta(t'))^T|\bx_0]
=\big(\tau(t)\tau(t')\int_0^t\frac{\sigma(s)^2}{\tau(s)^2}ds\big)\bI_d
=\kappa(t,t')\bI_d.
\end{eqnarray}

Importantly, the expression in (\ref{affine}) subsumes all three reference SDEs proposed in \citet{song2021scorebased}, which have been successfully employed in a variety of probabilistic generative modeling tasks:
\begin{itemize}
\item Variance Exploding (VE) SDE:
\begin{eqnarray}\nn
c(t)=0,~~\alpha(t)=0.
\end{eqnarray}
When $\sigma(t)=\sigma$, this corresponds to the standard $d$-dimensional Brownian motion.
\item Variance Preserving (VP) SDE:
\begin{eqnarray}\label{vpsde}
c(t)=-\tfrac{1}{2}\beta(t),\alpha(t)=0,\sigma(t)=\sqrt{\beta(t)},
\end{eqnarray}
where $\beta(\cdot):[0,1]\rightarrow\bR_+$. When $\beta(t)=1$, the process reduces to the Ornstein–Uhlenbeck process.
\item Sub-Variance Preserving (sub-VP) SDE:
\begin{eqnarray}\nn
c(t)=-\tfrac{1}{2}\beta(t),\alpha(t)=0,\sigma(t)=\sqrt{\beta(t)\left(1-e^{-2\int_0^t\beta(s),ds}\right)}.
\end{eqnarray}
\end{itemize}
The VE SDE produces trajectories with diverging variance as $t\to\infty$, whereas the VP SDE yields processes with bounded variance. In particular, if the initial distribution $\bx_0$ has unit variance, the VP SDE preserves this property for all $t\in[0,\infty)$. The two principal classes of score-based generative models—score matching with Langevin dynamics (SMLD) and denoising diffusion probabilistic models (DDPM)—can be viewed as discrete approximations of the VE and VP SDEs, respectively \citep{song2021scorebased}.

The following corollary provides the analytic solution to the SB problem between Gaussian measures.
\begin{corollary}\label{coro}.
Let $\mu_0=\mathcal{N}(\nu_0,\bSigma_0)$ and $\mu_1=\mathcal{N}(\nu_1,\bSigma_1)$ be two arbitrary Gaussian distributions. Then, for the reference measure defined by (\ref{affine}), the closed-form expression of (\ref{drift0}) is given by
\begin{eqnarray}\nn
\bu^\star(\bx,t)&=&\sigma(t)^2\left(\frac{\tau(1)}{\tau(t)\kappa(1,1)}[r(t)\bSigma_1+\bar{r}(t)C_{\sigma_\star}^T]
-\frac{\tau(1)^2}{\tau(t)\kappa(1,1)}[\bar{r}(t)\bSigma_0+r(t)C_{\sigma_\star}]
-\rho(t)\bI_d\right)\\\nn
&&\bSigma_t^{-1}\left\{\bx-\bar{r}(t)\bnu_0-r(t)\bnu_1-\zeta(t)+r(t)\zeta(1)\right\}\\\label{exact}
&&+\frac{\sigma(t)^2\tau(1)}{\tau(t)\kappa(1,1)}(\bnu_1-\zeta(1))-\frac{\sigma(t)^2\tau(1)^2}{\tau(t)\kappa(1,1)}\bnu_0,
\end{eqnarray}
where $\sigma_\star^2=\kappa(1,1)/\tau(1)$,
\begin{eqnarray}\nn
D_\sigma=(4\bSigma_0^{1/2}\bSigma_1\bSigma_0^{1/2}+\sigma^4\bI_d)^{1/2},&&C_\sigma=\frac{1}{2}(\bSigma_0^{1/2}D_\sigma\bSigma_0^{-1/2}-\sigma^2\bI_d),
\end{eqnarray}
\begin{eqnarray}\nn
\bSigma_t=\bI_d+\frac{\tau(t)^2\kappa_1(t,t)}{\kappa(t,t)\kappa(1,1)}\bSigma_0
+\frac{\tau(1)^2\kappa(t,t)}{\tau(t)^2\kappa_1(t,t)\kappa(1,1)}\bSigma_1
+\frac{\tau(1)}{\kappa(1,1)}(C_{\sigma_\star}+C_{\sigma_\star}^T),
\end{eqnarray}
and the functions $\tau(t),\zeta(t),\kappa(t,t),\kappa_1(t,t),r(t),\bar{r}(t),\rho(t)$ are all defined in Appendix \ref{Gaussian}.
\end{corollary}
The complete derivation is presented in Appendix~\ref{Gaussian}, where it is verified that (\ref{exact}) exactly matches the well-known result in Theorem~3 of \citet{bunne2023schrodingerbridgegaussianmeasures}. Although the algebraic computation is somewhat tedious, the derivation itself is conceptually straightforward, as all transition kernels and the joint distribution $\mu^\star(\bx_0,\bx_1)$ in (\ref{drift0}) are multivariate Gaussians. In contrast, the proof in \citet{bunne2023schrodingerbridgegaussianmeasures} relies on considerably more sophisticated tools from entropic optimal transport, Riemannian geometry, and generator theory.

\section{Practical Implementation}\label{setting}

In this section, we detail the implementation of our closed-form framework described in Theorem~\ref{thm} for learning the stochastic differential equations (SDEs) of the Schrödinger Bridge (SB) problem from marginal observations. In Section~\ref{setting1}, we introduce a training-free algorithm derived directly from the theorem. In Section~\ref{setting2}, we present an equivalent variational formulation of Theorem~\ref{thm}, from which we derive a simulation-free algorithm that leverages deep neural network training.

\subsection{Training-free algorithm}\label{setting1}

The drift term (\ref{drift0}) involves integration over the joint distribution $\mu^\star(\bx_0,\bx_1)$, which typically lacks a closed-form expression. To evaluate (\ref{drift0}) accurately, one must estimate the probability density $\hat{\mu}^\star(\bx_0,\bx_1)$ from available samples. This estimation is particularly challenging in high-dimensional settings, as $\mu^\star$ is often multi-modal or concentrated on a low-dimensional manifold, making it difficult to approximate using simple distributional families such as Gaussians or Gaussian mixtures. Classical density estimation methods include kernel-based techniques \citep{10.1214/aoms/1177704472}, interpolation-based methods \citep{10.1214/aos/1176342998}, and diffusion-based approaches \citep{1056736}. However, these techniques are primarily effective in low-dimensional settings and do not scale efficiently to higher dimensions.

In this work, we exploit two key properties of formula (\ref{drift0}) to construct a direct estimator for the drift coefficient $\bu^\star(\bx,t)$.
First, since (\ref{drift0}) involves only expectations under $\mu^\star$, these integrals can be approximated by empirical averages:
\begin{eqnarray}\label{clt}
\int g_t(\bx,\bx_0,\bx_1)\mu(d\bx_0,d\bx_1)\approx\frac{1}{n}\sum_{i=1}^n g_t(\bx,\bx_0^{(i)},\bx_1^{(i)}),
\end{eqnarray}
which leads to the following estimator:
\begin{eqnarray}\label{drift}
\hat{\bu}^\star(\bx,t)=\frac{\sigma(t)^2\sum_{i=1}^n \bs_{\bx_1^{(i)}}(\bx,t) g_t(\bx,\bx_0^{(i)},\bx_1^{(i)})}{\sum_{i=1}^n g_t(\bx,\bx_0^{(i)},\bx_1^{(i)})}.
\end{eqnarray}
This formulation completely avoids neural network training, thereby eliminating architectural design considerations and substantially reducing computational cost. The second advantage arises from the fact that the transition density $q(t_1,\bx,t_2,\by)$ in (\ref{transition}) is defined with respect to the reference process, which can often be computed in closed form for suitable choices of the drift $\bb(\bx,t)$ and diffusion $\sigma(t)$ in (\ref{sde0}). For example, when $\bb(\bx,t)$ is affine in $\bx$, as in (\ref{affine}), the transition kernel becomes Gaussian.

The joint distribution $\hat{\mu}^\star$ in (\ref{drift0}) corresponds to the solution of the static Schrödinger Bridge problem (\ref{static}), which is equivalent to the entropically regularized optimal transport (EOT) problem:
\begin{eqnarray}\label{eot}
\mu^\star=\text{argmin}_{\mu\in\Pi(\mu_0,\mu_1)}\left\{\int\|\bx_1-(\tau(1)\bx_0+\zeta(1))\|^2d\mu+2\kappa(1,1)\int\log\mu, d\mu\right\},
\end{eqnarray}
under the reference SDE (\ref{affine}). Given paired data $\left\{\bx_0^{(i)},\bx_1^{(i)}\right\}_{i=1}^n$, we can minimize the empirical objective
\begin{eqnarray}\nn
\sum_{i=1}^n\|\bx_1^{(i)}-(\tau(1)\bx_0^{(i)}+\zeta(1))\|^2,
\end{eqnarray}
to estimate the parameters $\hat{\tau}(1)$ and $\hat{\zeta}(1)$, which in turn define an optimal reference SDE process. Once the reference SDE is specified, the corresponding function $g_t$ in (\ref{drift0}) can be substituted to obtain the drift estimator in (\ref{drift}).

In practice, paired samples are rarely available. Instead, we often have access only to i.i.d. samples forming two empirical marginal distributions $\hat{\mu}_0$ and $\hat{\mu}_1$, of size $m$ and $n$. In this case, one can obtain synthetic paired samples by solving the EOT problem (\ref{eot}) between $\hat{\mu}_0$ and $\hat{\mu}_1$. Importantly, empirical approximations of EOT remain accurate even in high-dimensional spaces \citep{aude2016stochasticoptimizationlargescaleoptimal}, and recent results show that the Schrödinger Bridge inherits this favorable property \citep{pmlr-v206-stromme23a}. The EOT plan between empirical distributions can be efficiently computed using the Sinkhorn algorithm \citep{NIPS2013_af21d0c9}, which has ${\cal O}(n^2)$ complexity \citep{NIPS2017_491442df}, or via stochastic optimization methods \citep{aude2016stochasticoptimizationlargescaleoptimal}. The complete training-free Schrödinger Bridge (TFSB) estimation procedure is summarized in Algorithm~\ref{alg1}.
\begin{algorithm}[H]
\caption{Training-free estimation}
\label{alg1}
\begin{algorithmic}
\STATE \textbf{Input:} data $\left\{\bx_0^{(i)}\right\}_{i=1}^m$ and $\left\{\bx_1^{(j)}\right\}_{j=1}^n$
\IF{paired data $\left\{\bx_0^{(i)},\bx_1^{(i)}\right\}_{i=1}^n$ are available}
\STATE Compute $\bu^\star(\bx,t)$ according to (\ref{drift})
\ELSE
\STATE Compute the EOT plan $\pi\leftarrow\text{Sinkhorn}\{\bx_0,\bx_1,2\kappa(1,1)\}$ using (\ref{eot})
\STATE Sample paired data $\left\{\bx_0^{(i)},\bx_1^{(i)}\right\}_{i=1}^n\sim\pi$
\STATE Compute $\bu^\star(\bx,t)$ according to (\ref{drift})
\ENDIF
\end{algorithmic}
\end{algorithm}

\subsection{Simulation-free deep neural network training algorithm}\label{setting2}

Leveraging insights from score-based generative learning, we introduce the following theorem, which provides an alternative characterization of the drift term $\bu^\star(\bx,t)$ as a conditional expectation. This formulation naturally leads to a quadratic objective function for which $\bu^\star(\bx,t)$ emerges as the unique minimizer.
\begin{theorem}\label{thm1}
The drift term $\bu^\star(\bx_t,t)$ in (\ref{sde}) is the time-varying vector field $\bu(\bx_t,t)$ that minimizes the quadratic objective
\begin{eqnarray}\label{lse01}
\mE_{\bQ}\left\|\sigma(t)^2\bs_{\bx_1}(t,\bx_t)-\bu(\bx_t,t)\right\|^2,
\end{eqnarray}
where $\bQ = [t \sim \mathcal{U}(0,1)] \otimes \mu^\star(\bx_0,\bx_1) \otimes \pi(\bx_t|\bx_0,\bx_1)$, the conditional score is defined in (\ref{score}), and the conditional bridge distribution is given by
\begin{eqnarray}\label{bridge}
\pi(\bx_t|\bx_0,\bx_1) = \frac{q(0,\bx_0,t,\bx_t)\, q(t,\bx_t,1,\bx_1)}{q(0,\bx_0,1,\bx_1)}.
\end{eqnarray}
\end{theorem}
Equation~(\ref{lse01}) provides a practical foundation for designing deep learning algorithms to estimate $\bu^\star(\bx_t,t)$ nonparametrically. In particular, one can approximate $\bu^\star(\bx,t)$ by training a neural network $\bu_{\btheta}(\bx,t)$ on i.i.d.\ samples $t_i \sim \mathcal{U}(0,1)$, $(\bx_{0,i},\bx_{1,i}) \sim \mu^\star(\bx_0,\bx_1)$, and $\bx_{t,i} \sim \pi(\bx_{t,i} | \bx_{0,i},\bx_{1,i})$. The resulting regression problem seeks to minimize the expected squared error in (\ref{lse01}), thereby learning the time-dependent drift $\bu^\star$ as a function of $(\bx, t)$ given joint samples $(\bx_0, \bx_1) \sim \mu^\star(\bx_0,\bx_1)$. In this way, the estimation task reduces to a supervised learning problem for the drift field.

It is worth noting that Theorem~\ref{thm1} encompasses several existing simulation-free estimation methods for the Schrödinger Bridge problem as special cases. For instance, \citet{tong2024} propose a simulation-free score and flow matching framework for learning stochastic dynamics from unpaired samples drawn from arbitrary source and target distributions. Their method introduces two neural networks to match the conditional ODE flow $\bu^o(\bx,t|\bx_0,\bx_1)$ and the conditional score $\nabla_{\bx}\log p_t(\bx|\bx_0,\bx_1)$ separately, after which the conditional drift is reconstructed as
\begin{equation}\nn
\bu(\bx,t|\bx_0,\bx_1)
= \bu^o(\bx,t|\bx_0,\bx_1)
+ \frac{\sigma(t)^2}{2}\nabla_{\bx}\log p_t(\bx|\bx_0,\bx_1).
\end{equation}
From the explicit expressions\footnote{There is a typo in Eq.~(8) of \citet{tong2024}, where the first term in the first equation should be multiplied by $1/2$.} of $\bu^o(\bx,t|\bx_0,\bx_1)$ and $\nabla_{\bx}\log p_t(\bx|\bx_0,\bx_1)$ for Brownian bridges given in (8) of \citet{tong2024}, we obtain
\begin{equation}\nn
\bu(\bx,t|\bx_0,\bx_1) = \frac{\bx_1 - \bx}{1 - t},
\end{equation}
which coincides exactly with the $\sigma(t)^2 \bs_{\bx_1}(t,\bx_t)$ term in (\ref{lse01}) when the reference process is Brownian motion. Hence, the method of \citet{tong2024} can be viewed as a special instance of our framework, with the additional advantage that our approach requires fitting only a single neural network—substantially reducing computational cost.

Similarly, \citet{liu2023schrodinger} introduce the Image-to-Image Schrödinger Bridge (I$^2$SB) model, a family of conditional diffusion models that directly learn nonlinear diffusion processes between two prescribed distributions. While I$^2$SB shares conceptual similarities with score-based generative models and is effective for conditional sample generation, it does not yield an exact estimation of the underlying SDE drift term. Moreover, I$^2$SB is limited to the Brownian-motion reference process.

The regression-based training objective in (\ref{lse01}) provides a principled and efficient way to estimate the vector field $\bu(\bx, t)$ using only joint samples from $\mu^\star(\bx_0, \bx_1)$ and conditional samples from the tractable path distribution $\pi(\bx_t | \bx_0,\bx_1)$. We define the empirical loss as
\begin{eqnarray}\label{lhat}
\hat{\cal L}(\btheta)
= \frac{1}{mn}\sum_{i=1}^n\sum_{j=1}^m
\left\|
\sigma(t_j)^2\bs_{\bx_1^{(i)}}(t_j,\bx^{(i)}_{t_j})
- \bu_{\btheta}(\bx^{(i)}_{t_j},t_j)
\right\|^2,
\end{eqnarray}
where $(\bx_0^{(i)},\bx_1^{(i)})\sim p_{\text{data}}$,  $t_j\sim{\cal U}[\epsilon,1-\epsilon]$, and $\bx_{t_j}^{(i)}\sim\pi(\bx_{t_j}|\bx_0^{(i)},\bx_1^{(i)})$. To avoid numerical instability near the boundaries $t = 0$ and $t = 1$, we truncate the unit-time interval by $\epsilon$ at both ends, a common practice in score-based models \citep{song2021scorebased}.

We parameterize $\bu^\star(\bx,t)$ in (\ref{lse01}) by a neural network $\bu_{\btheta}(\bx,t)$ with trainable parameters $\btheta$, implemented as a feed-forward network expressive enough to approximate the true drift. The overall simulation free Schrödinger Bridge (SFSB) training process is summarized in Algorithm~\ref{alg2}.

\begin{algorithm}[H]
\caption{Simulation free training process of SBCG}
\label{alg2}
\begin{algorithmic}[1]
\STATE \textbf{Input:} Paired data $p_{\text{data}}=\{(\bx_0^{(i)},\bx_1^{(i)})\}_{i=1}^n$, number of iterations $L$.
\FOR{$l=1,\ldots,L$}
    \STATE Sample $(\bx_0^{(i)},\bx_1^{(i)})\sim p_{\text{data}}$.
    \STATE Sample $t_j\sim{\cal U}[\epsilon,1-\epsilon]$.
    \STATE Sample $\bx^{(i)}_{t_j}\sim\pi(\bx_{t_j}|\bx_0^{(i)},\bx_1^{(i)})$.
    \STATE Compute $\hat{\cal L}(\btheta)$ using (\ref{lhat}).
    \STATE Update $\btheta \leftarrow \btheta - \eta \nabla_{\btheta}\hat{\cal L}(\btheta)$.
\ENDFOR
\STATE \textbf{Output:} Trained drift model $\bu_{\btheta}(\bx,t)$.
\end{algorithmic}
\end{algorithm}

Note that our framework is not limited to neural networks—$\bu_{\btheta}$ may also be represented using other function classes such as kernel methods. However, for high-dimensional data and large-scale problems, neural networks often exhibit superior performance due to their expressive power.

\section{Numerical experiments}\label{numerical}

In this section, we demonstrate the effectiveness of our algorithm on both simulated and real datasets. We begin with two-dimensional toy examples to illustrate the ability of our method to learn multimodal distributions. We then show how our approach can be adapted for modeling single-cell dynamics. Finally, we validate our method on a range of image restoration tasks—including inpainting, super-resolution, deblurring, and JPEG restoration—using the CIFAR-10 datasets. 

Once the drift function $\bu^\star(\bx,t)$ is obtained, the diffusion process (\ref{sde}) can be used to generate samples from the target distribution $\mu_1$ by transporting the initial distribution $\mu_0$ at $t=0$ to $\mu_1$ at $t=1$. Since the SDE (\ref{sde}) is defined over the finite interval $[0,1]$, we implement this procedure numerically using the Euler–Maruyama discretization. Specifically, we divide $[\epsilon,1-\epsilon]$ into $N \geq 2$ grid points, $\epsilon = t_0 < t_1 < \cdots < t_N = 1-\epsilon$ with step sizes $\delta_j = t_{j+1} - t_j$. To mitigate numerical instability near the boundaries $t=0$ and $t=1$, we truncate the interval using a small $\epsilon > 0$. The discretized update rule is given by
\begin{eqnarray}\label{euler}
\bx_{t_{j+1}} = \bx_{t_j} + \delta_j[\bb(\bx_{t_j},t_j) + \bu^\star(\bx_{t_j},t_j)] + \sigma(t_j)\sqrt{\delta_j}\bepsilon_j, \quad j = 0,1,\ldots,N-1,
\end{eqnarray}
where $\{\bepsilon_j\}_{j=1}^N$ are i.i.d. random vectors drawn from $N(0,\bI_d)$. The drift term $\bu^\star(\bx_{t_j},t_j)$ is estimated using either Algorithm~\ref{alg1} or Algorithm~\ref{alg2}.

\subsection{Learning low-dimensional distributions}\label{lowdim}

We first evaluate the capability of various methods to learn multi-modal distributions on low-dimensional datasets. Following the setup in \cite{shi2023diffusion,tong2024}, we consider two widely used benchmark examples to illustrate and compare different approaches. The first is the Moons dataset, sampled from a complex distribution supported on two disjoint, half-moon–shaped regions of equal mass \citep{pedregosa11a}. The second is the 8-Gaussians dataset, generated from a mixture of eight Gaussian components that are well separated in space. 

Consistent with the experimental configurations in \cite{tong2024}, we use 10,000 samples for training and another 10,000 for testing. We set the initial distribution $\mu_0$ and use $N = 100$ discretization steps. Each experiment is repeated five times, and we report both the mean and standard deviation across runs. Figure~\ref{figure0} shows samples generated by our TFSB method on the 8-Gaussians dataset, using standard normal (left) and Moons (right) starting distributions. The generated samples closely match the true distributions, successfully capturing all modes in both cases.
\begin{figure}[ht!]
    \vspace{-0.5cm}
		\begin{minipage}[t]{0.47\linewidth}
\centering
\includegraphics[width=1.08\textwidth]{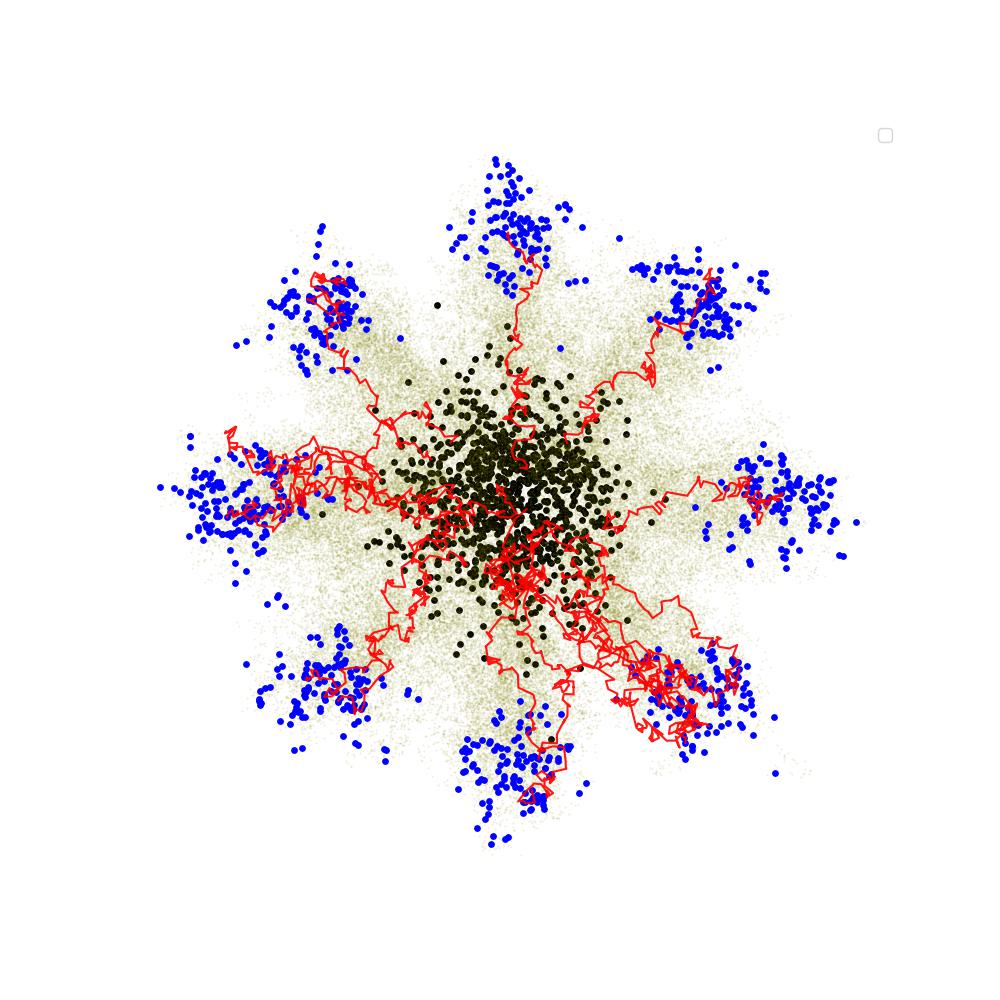}
		\end{minipage}
		\hspace{0.1cm}
		\begin{minipage}[t]{0.47\linewidth}
\centering
\includegraphics[width=1.08\textwidth]{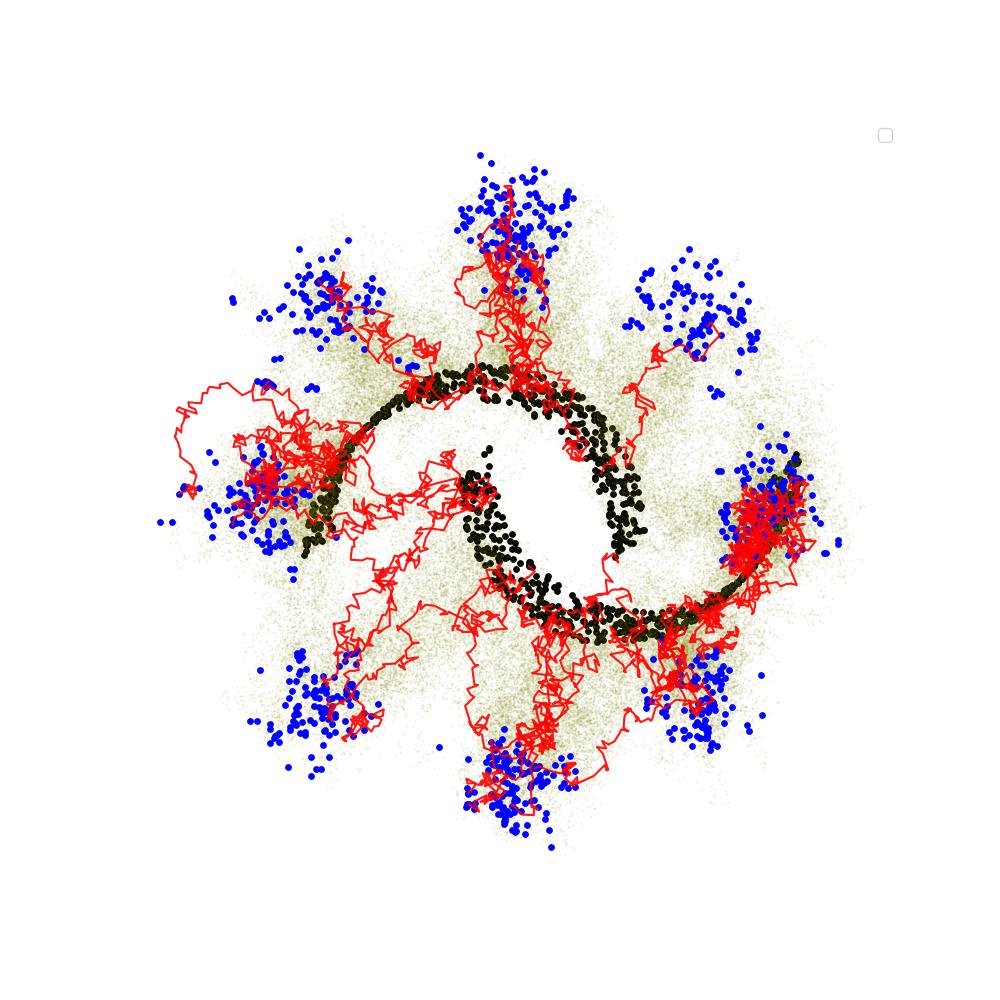}
		\end{minipage}
    \vspace{-0.5cm}
\caption{Visualization of the 8-Gaussians data generated from the standard normal (left) and the Moons dataset (right). Black points denote the starting samples, blue points represent the target samples generated by our TFSB method, and red curves illustrate the entire transport paths.}
\label{figure0}
\end{figure}

We evaluate the empirical 2-Wasserstein distance for 10,000 samples generated by our model. The reported value of ${\cal W}_2$ is defined as
\begin{eqnarray}\nn
{\cal W}2 = \left(\min_{\pi \in U(\hat{p}_1, q_1)} \int |\bx - \by|_2^2 d\pi(\bx, \by)\right)^{1/2},
\end{eqnarray}
where $\hat{p}_1$ denotes the samples generated by our algorithm, $q_1$ is the test set, and $U(\hat{p}_1, q_1)$ is the set of all couplings between $\hat{p}_1$ and $q_1$. We report the 2-Wasserstein distance between the predicted and target distributions, each based on 10,000 test samples. Table~\ref{table1} summarizes the comparison of our methods with the following baselines under three scenarios: ${\cal N}\rightarrow$8-Gaussians, ${\cal N}\rightarrow$Moons, and Moons$\rightarrow$8-Gaussians.
\begin{itemize}
\item Simulation-Free Schrödinger Bridges via Score and Flow Matching \citep{tong2024}, including the versions with exact optimal transport couplings ($[SF]^2$M-exact) and independent couplings ($[SF]^2$M-I);
\item Iterative Schrödinger Bridge Models: diffusion Schrödinger bridges (DSB), and diffusion Schrödinger bridge matching with algorithms based on iterative proportional fitting (DSBM-IPF) and iterative Markovian fitting (DSBM-IMF);
\item ODE Flow-Based Models: optimal transport conditional flow matching (OT-CFM), Schrödinger bridge conditional flow matching (SB-CFM), independent conditional flow matching (I-CFM), rectified flow (RF), and flow matching (FM).
\end{itemize}

Results obtained from SDE-based stochastic methods and ODE-based deterministic methods are presented in the top and bottom blocks of Table~\ref{table1}, respectively, following the official values reported by each baseline in \cite{tong2024}. Across all scenarios, our TFSB method consistently outperforms the SFSB one. For the Moons$\rightarrow$8-Gaussians task, TFSB achieves the best overall performance with a large margin over all competing methods. In the other two settings, TFSB ranks third, performing very close to the top method within the confidence interval. Overall, the results in Table~\ref{table1} demonstrate that TFSB is a highly competitive generative model for low-dimensional data.
\begin{table}
  \begin{center}
    \caption{Summary of generative modeling performance of various methods on two-dimensional datasets, evaluated using the 2-Wasserstein distance (${\cal W}_2$). Reported values are the mean and standard deviation computed over five independent runs.}
    \label{table1}
    \vspace{0.35cm}
    \begin{tabular}{c|c|c|c}\hline
      Algorithm&${\cal N}\rightarrow$8-Gaussians&${\cal N}\rightarrow$Moons&Moons$\rightarrow$8-Gaussians\\\hline
      TFSB&0.314\std{0.045}&0.138\std{0.013}&\textbf{0.350\std{0.029}}\\
      SFSB&0.632\std{0.046}&0.379\std{0.057}&0.590\std{0.065}\\
      $[SF]^2$M-exact \citep{tong2024}&\textbf{0.275\std{0.058}}&\textbf{0.124\std{0.023}}&0.726\std{0.137}\\
      $[SF]^2$M-I \citep{tong2024}&0.393\std{0.054}&0.185\std{0.028}&1.482\std{0.151}\\
      DSBM-IPF \citep{shi2023diffusion}&0.315\std{0.079}&0.140\std{0.006}&0.812\std{0.092}\\
      DSBM-IMF \citep{shi2023diffusion}&0.338\std{0.091}&0.144\std{0.024}&0.838\std{0.098}\\
      DSB \citep{debortoli2023diffusion}&0.411\std{0.084}&0.190\std{0.049}&0.987\std{0.324}\\\hline
      OT-CFM \citep{tong2024improving}&0.303\std{0.053}&0.130\std{0.016}&0.601\std{0.027}\\
      SB-CFM \citep{tong2024improving}&2.314\std{2.112}&0.434\std{0.594}&---\\
      RF \citep{liu2022rectified}&0.421\std{0.071}&0.283\std{0.045}&1.525\std{0.330}\\
      I-CFM \citep{tong2024improving}&0.373\std{0.103}&0.178\std{0.014}&1.557\std{0.407}\\
      FM \citep{lipman2023flow}&0.343\std{0.058}&0.209\std{0.055}&---\\\hline
    \end{tabular}
\end{center}
  \vspace{-0.25cm}
\end{table}

\subsection{Learning Cell Dynamics}\label{set2}

In this section, we demonstrate how TTSB and SFSB can be applied to modeling single-cell dynamics, a central problem in single-cell data science. Understanding these dynamics is key to characterizing, and ultimately manipulating, cellular programs underlying development and disease~\citep{lahnemann_eleven_2020}. Given time-resolved snapshot data—where each observation corresponds to a cell embedded in a gene expression state space—our goal is to infer the short-timescale transitions that drive cellular state changes~\citep{schiebinger_reconstructing_2021}.

We evaluate our methods on three real-world datasets following the setup of \citet{tong_trajectorynet_2020}, reporting results across data representations of varying dimensionalities in Tables~\ref{tab:sc-dynamics-fived} and~\ref{tab:sc-dynamics}. The CITE-seq and Multiome datasets, both from the NeurIPS 2022 Multimodal Single-cell Integration Challenge~\citep{burkhardt_multimodal_2022}, consist of CD34+ hematopoietic stem and progenitor cells measured at four timepoints (days 2, 3, 4, and 7). The Embryoid Body (EB) dataset contains five timepoints collected over 30 days; we use the preprocessed version containing the first 100 principal components~\citep{tong_trajectorynet_2020}.

Given $K$ unpaired cell distributions at $K$ timepoints, we solve a Schrödinger bridge problem between each successive pair. To evaluate interpolation, we perform leave-one-out prediction: for each omitted timepoint, we train on the remaining data, push forward the observed distribution at time $t-1$ to $t$, and compute the 1-Wasserstein distance to the held-out distribution.

We compare performance using the first 5, 50, or 100 whitened principal components. In the low-dimensional (5d) setting, TTSB and SFSB match the performance of the ODE-based OT-CFM and outperform other Schrödinger bridge approaches. Their advantage becomes more pronounced in higher dimensions, where they significantly outperform all baselines and remain stable while alternative methods degrade. All evaluations use the 1-Wasserstein distance between predicted and ground-truth cell populations.

\begin{table}[t]
  \begin{center}
\caption{Comparison of single-cell dynamics across three datasets using the 5-dimensional PCA representation. For each dataset, we leave out each intermediate timepoint in turn, impute its distribution by pushing forward the preceding timepoint, and compute the 1-Wasserstein distance to the held-out distribution, following \citet{tong_trajectorynet_2020}. }\vspace*{1em}
\label{tab:sc-dynamics-fived}
\begin{tabular}{@{}llll}
\toprule
Algorithm  $\downarrow$ | Dataset $\rightarrow$ &                       \multicolumn1c{Cite}&                          \multicolumn1c{EB} &                      \multicolumn1c{Multi}\\
\midrule
TFSB&0.882\std{0.020}&\textbf{0.791\std{0.014}}&1.017\std{0.014}\\
SFSB&\textbf{0.880\std{0.013}}&0.792\std{0.016}&1.029\std{0.020}\\
{[SF]}$^2$M-Geo&1.017\std{0.104}&0.879\std{0.148}&1.255\std{0.179}\\
{[SF]}$^2$M-Exact&0.920\std{0.049}&0.793\std{0.066}&\textbf{0.933\std{0.054}} \\
{[SF]}$^2$M-Sink&1.054\std{0.087}&1.198\std{0.342}&1.098\std{0.308} \\
DSBM \citep{shi2023diffusion}&1.705\std{0.160}&1.775\std{0.429}&1.873\std{0.631} \\
DSB \citep{debortoli2023diffusion}&0.953\std{0.140}&0.862\std{0.023}&1.079\std{0.117} \\
\midrule
OT-CFM \citep{tong2024improving}&\textbf{0.882\std{0.058}}&\textbf{0.790\std{0.068}}&\textbf{0.937\std{0.054}} \\
I-CFM \citep{tong2024improving}&0.965\std{0.111}&0.872\std{0.087}&1.085\std{0.099} \\
SB-CFM \citep{tong2024improving}&1.067\std{0.107}&1.221\std{0.380}&1.129\std{0.363} \\
\midrule
Reg. CNF \citep{finlay_how_2020}&---&0.825&--- \\
TrajectoryNet \citep{tong_trajectorynet_2020}&---&0.848&--- \\
NLSB \citep{koshizuka_neural_2022} & --- & 0.970 & ---\\
\bottomrule
\end{tabular}
\end{center}
\end{table}

\begin{table}[t]
 \caption{Leave-one-timepoint-out performance of dynamics interpolation methods. For each omitted timepoint, we predict its cell-state distribution from the preceding timepoint and compute the 1-Wasserstein distance to the ground truth. Results are reported for data represented using 50 and 100 principal components.}\vspace*{1em}
    \label{tab:sc-dynamics}
    \centering
\begin{tabular}{@{}llllll}
\toprule
 Dim. $\rightarrow$ &  \multicolumn{2}{c}{50} & \multicolumn{2}{c}{100}\\
 \cmidrule(lr){2-3}\cmidrule(lr){4-5}
 Alg.\ $\downarrow$ | Dataset $\rightarrow$ &                         \multicolumn1c{Cite} &                                                \multicolumn1c{Multi} &                         \multicolumn1c{Cite} &                                                 \multicolumn1c{Multi}\\
\midrule
TFSB                       & \textbf{6.35\std{0.01}} &    \textbf{6.49\std{0.01}} &  10.29\std{0.03} &  10.52\std{0.02} \\
SFSB                     & 6.45\std{0.02} &                  6.57\std{0.02} &           \textbf{10.24\std{0.03}} &                    \textbf{10.46\std{0.03}}\\
{[SF]}$^2$M-Geo                       & 38.52\std{0.29} &    44.80\std{1.91} &  44.50\std{0.42} & 52.20\std{1.96} \\
{[SF]}$^2$M-Exact                     & 40.01\std{0.78} &                  45.34\std{2.83} &           46.53\std{0.43} &                    52.89\std{1.99}\\
DSBM        & 53.81\std{7.74}  &  66.43\std{14.39} &  58.99\std{7.62} &    70.75\std{14.03}\\
\midrule
OT-CFM   & 38.76\std{0.40} &             47.58\std{6.62} &           45.39\std{0.42} &            54.81\std{5.86}\\
I-CFM    & 41.83\std{3.28} &    49.78\std{4.43} &  48.28\std{3.28} &  57.26\std{3.86} \\
\bottomrule
\end{tabular}
\end{table}

\subsection{Image restoration}\label{set3}

Image restoration is a fundamental problem in computer vision and image processing, with applications ranging from optimal filtering \citep{10165442} and data compression \citep{125072} to adversarial defense \citep{nie2022DiffPure} and safety-critical domains such as medical imaging and robotics \citep{LI2021106617}. Recent advances have leveraged diffusion and score-based generative models for restoration tasks; see, for example, \cite{saharia2022paletteimagetoimagediffusionmodels,kawar2022denoisingdiffusionrestorationmodels}. However, as noted by \cite{liu2023schrodinger}, most diffusion-based restoration methods initiate the generative process from Gaussian white noise, which contains little to no structural information about the underlying clean image distribution.

This observation highlights the advantage of Schrödinger bridge (SB) based diffusion models for image restoration: the degraded observation itself provides a structurally informative prior, in stark contrast to random noise. In particular, Theorem \ref{thm} shows that the SB dynamics admit analytic solutions when provided with boundary pairs (clean and degraded images). This leads directly to a training-free framework (TFSB) or a simulation-free framework (SFSB), both of which avoid the computational overhead typically associated with diffusion model training or trajectory simulation.

We evaluate our method on four standard restoration tasks—super-resolution, deblurring, inpainting, and JPEG artifact removal—using CIFAR-10 images of size $3 \times 32 \times 32$. For 4× super-resolution, we downsample images to $8 \times 8$ and then upsample them back to $32 \times 32$ so that the clean and degraded domains share the same resolution. Inpainting is performed using masks that remove the central 25\% of pixels. Deblurring is done using Gaussian blur. For JPEG artifact removal, we apply JPEG compression with a quality factor of 10.

Because TFSB directly constructs diffusion bridges between the clean and degraded domains, the resulting generative trajectories are more interpretable: the restored image emerges progressively from the degraded input, as illustrated in Figures \ref{figure1} and \ref{figure2}. This structural guidance also improves sampling efficiency, since the starting point of generation is already close to the target image manifold. Figures \ref{figure1} and \ref{figure2} visualize the full restoration trajectories obtained using TFSB under the reference VP-SDE (\ref{vpsde}), demonstrating successful recovery of the target images.

Finally, all experiments were conducted on a consumer-grade laptop without access to a dedicated GPU. This prevents direct comparison with GPU-based methods or scaling to larger datasets such as ImageNet ($3 \times 256 \times 256$). With sufficient computational resources, we expect comparable performance to \cite{liu2023schrodinger}, as their simulation-free training procedure can be viewed as a special case of our more general framework.

\begin{figure}[ht!]
    \vspace{0.5cm}
		\begin{minipage}[t]{0.45\linewidth}
\centering
\includegraphics[width=1.15\textwidth]{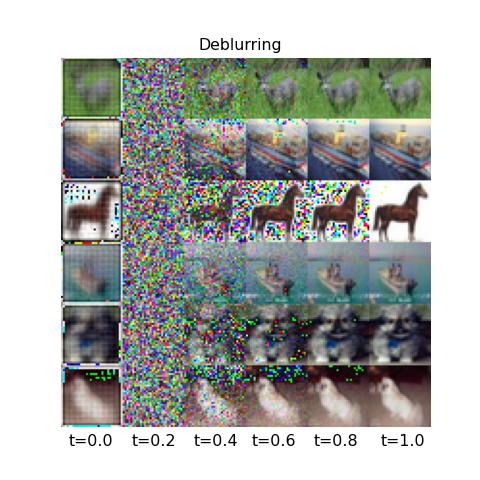}
		\end{minipage}
		\hspace{0cm}
		\begin{minipage}[t]{0.45\linewidth}
\centering
\includegraphics[width=1.15\textwidth]{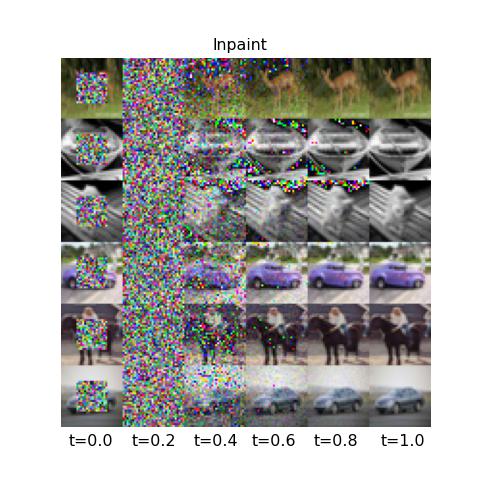}
		\end{minipage}
    \vspace{0.cm}
\caption{Visualization of the restoration trajectories generated by TFSB, demonstrating progressive evolution from degraded inputs to clean outputs. Left: Deblurring. Right: Inpainting.}
\label{figure1}
\end{figure}

\begin{figure}[ht!]
    \vspace{0.5cm}
		\begin{minipage}[t]{0.47\linewidth}
\centering
\includegraphics[width=1.15\textwidth]{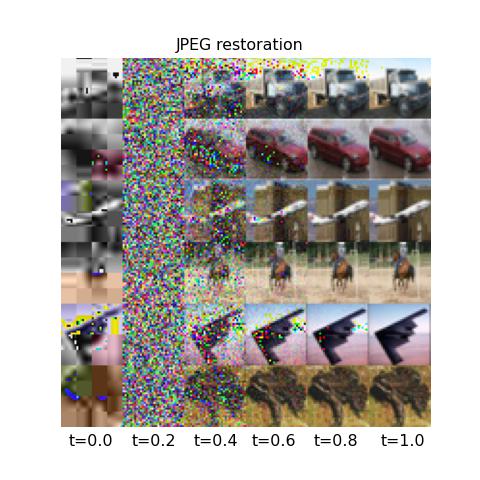}
		\end{minipage}
		\hspace{0.0cm}
		\begin{minipage}[t]{0.47\linewidth}
\centering
\includegraphics[width=1.15\textwidth]{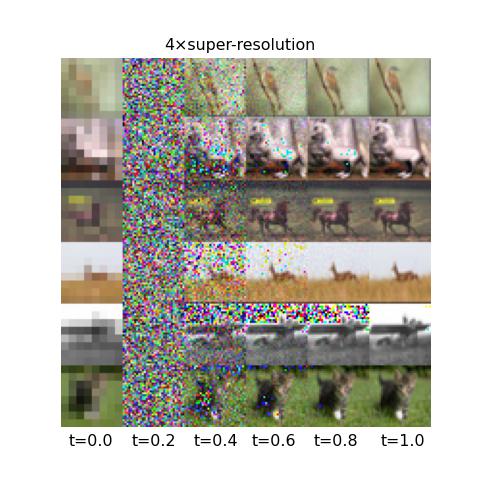}
		\end{minipage}
    \vspace{0.cm}
\caption{Visualization of the restoration trajectories generated by TFSB, demonstrating progressive evolution from degraded inputs to clean outputs. Left: JPEG artifact removal. Right: $4\times$ super-resolution.}
\label{figure2}
\end{figure}

\section{Conclusion}\label{conclusion}

In this work, we introduced a closed-form framework for learning the stochastic dynamics of Schrödinger bridge systems by deriving an explicit expression for the drift function of the associated SDE. This result provides a unified theoretical perspective that encompasses all previously known closed-form SB solutions. The formulation is notably simple, requiring only the transition probabilities of the reference process and integration with respect to the joint target distribution. Building on this insight, we developed a training-free algorithm (Algorithm 1) capable of learning stochastic dynamics from both paired and unpaired data. The method operates in a single step, avoiding neural network training and the associated challenges of architectural design, making it computationally efficient and straightforward to implement. Through an equivalent variational formulation, we further proposed a simulation-free procedure for inferring dynamics directly from samples drawn from arbitrary source and target distributions. These approaches prove particularly effective for image restoration and for modeling temporal evolution in single-cell genomic data. Experimental results on both low-dimensional multi-modal distributions and high-dimensional image datasets demonstrate that our method achieves competitive performance relative to state-of-the-art models while remaining simple, efficient, and computationally lightweight. Moreover, unlike previous simulation-free approaches which restrict the reference dynamics to Brownian motion, our framework accommodates general Itô reference processes.

A known limitation, however, is that the closed-form drift solution (\ref{drift}) tends to memorize samples, which constrains its ability to generate genuinely novel outputs \citep{scorebasedgenerativemodelsdetect}. To mitigate this issue, \cite{closedformdiffusionmodels} proposed smoothing the closed-form score using convolution with zero-mean noise together with efficient nearest-neighbor approximations. Their modification enables novel sample generation while retaining a training-free formulation. As future work, we intend to incorporate similar smoothing strategies into our framework to enhance generative capability. In addition, we plan to explore the integration of convolutional architectures from recent neural methods into our smoothing-based, training-free TFSB approach to further improve empirical performance.

Our formulation also suggests several promising extensions. One natural direction is to consider conditional settings in which the stochastic dynamics are learned between two marginal distributions conditioned on an auxiliary variable $\bz$. This can be achieved by introducing a neural function $\bu_{\btheta}(\bx,\bz,t)$ in (\ref{lhat}) to incorporate the conditional dependence, similar to conditional SGMs \citep{song2021scorebased} and conditional SB formulations \citep{HUANG2025105486}. Another direction is to generalize the approach to multi-marginal problems, following recent developments in multi-marginal Schrödinger bridges \citep{shen2025multimarginalschrodingerbridgesiterative}. We expect these avenues to broaden the applicability of our framework to a wider class of nonlinear diffusion models.

\appendix

\section{Static Formulation}\label{statica}
\begin{lemma}\label{lemma}
Let $\mu^\star$ denote the solution of the static Schrödinger Bridge problem, i.e. the joint distribution $\mu(\bx_0,\bx_1)$ on $\bR^d\times\bR^d$ that minimizes the Kullback–Leibler divergence to a given reference distribution $\bar{\mu}(\bx_0,\bx_1)$ subject to prescribed marginals $\mu_0$ and $\mu_1$. Formally,
\begin{equation}\label{opt}
\mu^\star=\text{argmin}_{\mu\in\Pi(\mu_0,\mu_1)} D(\mu|\tilde\mu),
\end{equation}
where $\Pi(\mu_0,\mu_1)$ is the set of couplings with marginals $\mu_0$ and $\mu_1$. Assume the reference transition density $q(0,\bx_0,1,\bx_1)$ is positive on the relevant domain. Then there exist strictly positive functions $\hat{\rho}_0,\rho_1:\bR^d\rightarrow(0,\infty)$ such that the optimal coupling factors as
\begin{equation}\label{factor}
\mu^\star(\bx_0,\bx_1)
= q(0,\bx_0,1,\bx_1)\hat{\rho}_0(\bx_0)\rho_1(\bx_1).
\end{equation}
Moreover, $\mu^\star$attains the prescribed marginals, which can be written in the form
\begin{equation}\label{marginal}
\mu_0(\bx_0)=\rho_0(\bx_0)\hat{\rho}_0(\bx_0),
\qquad
\mu_1(\bx_1)=\rho_1(\bx_1)\hat{\rho}_1(\bx_1),
\end{equation}
with the consistency relations
\begin{eqnarray}\nn
\rho_0(\bx_0)=\int q(0,\bx_0,1,\bx_1)\rho_1(\bx_1)d\bx_1,&&\hat{\rho}_1(\bx_1)=\int q(0,\bx_0,1,\bx_1)\hat{\rho}_0(\bx_0)d\bx_0.
\end{eqnarray}
\end{lemma}

\begin{proof}
The optimization problem (\ref{opt}) can be written as:
\begin{eqnarray}\label{ot}
\mu^\star(\bx_0,\bx_1)=\text{argmin}_{\mu\in\Pi(\mu_0,\mu_1)}\int\left\{\log[\mu(\bx_0,\bx_1)]-\log[q(0,\bx_0,1,\bx_1)]\right\}\mu(d\bx_0,d\bx_1).
\end{eqnarray}
Introduce the constrained minimization \eqref{ot}  by adjoining dual potentials (Lagrange multipliers) $\lambda_0(\bx_0)$ and $\lambda_1(\bx_1)$ to form the Lagrangian by
\begin{eqnarray}\nn
{\cal L}(\mu,\mu_0,\mu_1)&=&\int\left\{\log[\mu(\bx_0,\bx_1)]-\log[\bar{\mu}(\bx_0,\bx_1)]\right\}\mu(d\bx_0,d\bx_1)\\\nn
&&+\int\lambda_0(\bx_0)\left(\int\mu(\bx_0,\bx_1)d\bx_1-\mu_0(\bx_0)\right)d\bx_0\\\label{lagrang}
&&+\int\lambda_1(\bx_1)\left(\int\mu(\bx_0,\bx_1)d\bx_0-\mu_1(\bx_1)\right)d\bx_1.
\end{eqnarray}
Then, differentiating the Lagrangian with respect to $\mu(\bx_0,\bx_1)$ yields the first-order optimality condition
\begin{eqnarray}\nn
\log\mu^\star(\bx_0,\bx_1)-\log[\bar{\mu}(\bx_0,\bx_1)]+1+\lambda_0(\bx_0)+\lambda_1(\bx_1)=0.
\end{eqnarray}
Rearranging, we get
\begin{eqnarray}\nn
\mu^\star(\bx_0,\bx_1)=\mu_0(\bx_0)q(0,\bx_0,1,\bx_1)\exp(-\lambda_0(\bx_0)-\lambda_1(\bx_1)-1),
\end{eqnarray}
which we can re-express in terms of the auxiliary potentials $\hat{\rho}_0,\rho_1$:
\begin{eqnarray}\nn
\mu^\star(\bx_0,\bx_1)=\hat{\rho}_0(\bx_0)q(0,\bx_0,1,\bx_1)\rho_1(\bx_1),
\end{eqnarray}
satisfying
\begin{eqnarray}\nn
\hat{\rho}_0(\bx_0)\int q(0,\bx_0,1,\bx_1)\rho_1(\bx_1)d\bx_1=\mu_0(\bx_0),\\\nn
\rho_1(\bx_1)\int q(0,\bx_0,1,\bx_1)\hat{\rho}_0(\bx_0)d\bx_0=\mu_1(\bx_1),
\end{eqnarray}
where we re-label the terms with the integrals to
\begin{eqnarray}\nn
\rho_0(\bx_0)=\int q(0,\bx_0,1,\bx_1)\rho_1(\bx_1)d\bx_1,\\\nn
\hat{\rho}_1(\bx_1)=\int q(0,\bx_0,1,\bx_1)\hat{\rho}_0(\bx_0)d\bx_0.
\end{eqnarray}
These relations are equivalently written as
\begin{eqnarray}\nn
\hat{\rho}_0(\bx_0)\rho_0(\bx_0)=\mu_0(\bx_0),\\\nn
\hat{\rho}_1(\bx_1)\rho_1(\bx_1)=\mu_1(\bx_1),
\end{eqnarray}
which is the Schrödinger system (or Sinkhorn scaling equations) for the pair $(\hat{\rho}_0,\rho_1)$.
\end{proof}

\medskip\noindent\textbf{Remark.}
The scaling functions $\hat{\rho}_0$ and $\rho_1$ are determined only up to the multiplicative normalization $(\hat{\rho}_0,\rho_1)\rightarrow(c\hat{\rho}_0,c^{-1}\rho_1)$ for $c>0$, which leaves the coupling $\mu^\star$ invariant. Given any choice of $\hat{\rho}_0$ and $\rho_1$, define for each $t\in[0,1]$ the functions
\begin{align}\label{rhot}
\rho_t(\bx)
&= \int q(t,\bx,1,\bx_1),\rho_1(\bx_1),d\bx_1,\\\label{rhot1}
\hat{\rho}_t(\bx)
&= \int q(0,\bx_0,t,\bx),\hat{\rho}_0(\bx_0),d\bx_0,
\end{align}
where $q(s,\bx_s,t,\bx_t)$ denotes the transition density of the reference process (\ref{sde0}). By standard properties of transition kernels, $\hat{\rho}_t(\bx_t)$ satisfies the forward Fokker–Planck equation associated with (\ref{sde0}):
\begin{eqnarray}\label{forward}
\partial_t \hat{\rho}_t(\bx_t)=-\nabla_{\bx_t}\cdot\{\bb(\bx_t,t)\hat{\rho}_t(\bx_t)\}+\frac{\sigma(t)^2}{2}\bigtriangleup \hat{\rho}_t(\bx_t),
\end{eqnarray}
while $\rho_t(\bx_t)$ atisfies the corresponding backward Kolmogorov (backward Fokker–Planck) equation:
\begin{eqnarray}\label{backward}
\partial_t\rho_t(\bx_t)=-\{\nabla_{\bx_t}\rho_t(\bx_t)\}\cdot\bb(\bx_t,t)-\frac{\sigma(t)^2}{2}\bigtriangleup\rho_t(\bx_t).
\end{eqnarray}
Set $p_t(\bx_t)=\hat{\rho}_t(\bx_t)\rho_t(\bx_t)$, A direct computation using \eqref{forward} and \eqref{backward} shows that $p_t(\bx_t)$ satisfies the forward Fokker–Planck equation
\begin{eqnarray}\nn
\partial_t p_t(\bx_t)=-\nabla_{\bx_t}\cdot[\{\bb(\bx_t,t)+\sigma(t)^2\nabla_{\bx_t}\log\rho_t(\bx_t)\}p_t(\bx_t)]+\frac{\sigma(t)^2}{2}\bigtriangleup p_t(\bx_t).
\end{eqnarray}
Consequently, $p_t$ is the time-$t$ marginal of the diffusion process governed by the stochastic differential equation
\begin{eqnarray}\label{sdes}
d\bx_t=[\bb(\bx_t,t)+\sigma(t)^2\nabla_{\bx_t}\log\rho_t(\bx_t)] dt+\sigma(t) d\bw_t,~~\bx_0=\mu_0.
\end{eqnarray}
In other words, the drift correction $\sigma(t)^2\nabla_{\bx_t}\log\rho_t(\bx_t)$ transforms the reference dynamics into the optimal dynamics whose marginals are given by $p_t=\hat{\rho}_t\rho_t$.

\section{Proof of Theorem \ref{thm}}\label{proof}

To prove Theorem \ref{thm}, we begin by reformulating the KL divergence $D(\mQ|\bP)$ so as to disentangle the contribution of the terminal distribution $\mu(\bx_0,\bx_1)$ from that of the path measures. This separation is made possible by the disintegration theorem, which decomposes the divergence into one part depending only on the terminal distributions and another part depending on the conditional path laws. We then relate this decomposition to the drift and diffusion terms of the SDEs in (\ref{sde0}) and (\ref{sde}) via Girsanov’s theorem.

Recall that $\mQ_{01}(\bx_0,\bx_1)$ denotes the terminal distribution of $\mQ$. Let $\mQ^{\bx_0\bx_1}=\mQ((\bx_t)_{t\in(0,1)}|\bx_0,\bx_1)$ denote the conditional law of the trajectory given the endpoint $\bx_0,\bx_1$ under $\mQ$. Then $\mQ((\bx_t)_{t\in[0,1]})=\mQ^{\bx_0\bx_1}\mQ_{01}(\bx_0,\bx_1)$. Analogously, under the reference measure $\bP$, we define $\bP^{\bx_0\bx_1}=\bP((\bx_t)_{t\in(0,1)}|\bx_0,\bx_1)$. If $\mQ\ll\bP$ (so $D(\mQ|\bP)<\infty$), the disintegration theorem yields
\begin{eqnarray}\label{kld}
D(\mQ|\bP)=D(\mQ_{01}|\bP_{01})+\int_{\bR^d\otimes\bR^d}\mQ_{01}(d\bx_0,d\bx_1)D(\mQ^{\bx_0\bx_1}|\bP^{\bx_0\bx_1}).
\end{eqnarray}

Denote by $\mu$ and $\bar{\mu}$ the terminal joint laws of $\mQ$ and $\bP$, respectively, i.e., $\mQ_{01}(\bx_0,\bx_1)=\mu(\bx_0,\bx_1)$ and $\bP_{01}(\bx_0,\bx_1)=\bar{\mu}(\bx_0,\bx_1)$. Since $\mQ \in \mathcal{D}(\mu_0,\mu_1)$, the marginals of $\mu$ are $\mu_0$ and $\mu_1$. The marginal distribution of $\bar{\mu}$ with respect to $\bx_0$ is $\mu_0$, and it admits the factorization $\bar{\mu}(\bx_0,\bx_1)=\mu_0(\bx_0)q(0,\bx_0,1,\bx_1)$. Substituting into (\ref{kld}) gives
\begin{eqnarray}\nn
D(\mQ|\bP)=D(\mu|\bar{\mu})+\int_{\bR^d\otimes\bR^d}\mu(d\bx_0,d\bx_1)D(\mQ^{\bx_0\bx_1}|\bP^{\bx_0\bx_1})\ge D(\mu|\bar{\mu}).
\end{eqnarray}
This inequality has a clear interpretation: the divergence between $\mQ$ and $\bP$ is always at least as large as the divergence between their terminal distributions. Equality holds only when $\mQ^{\bx_0\bx_1}=\bP^{\bx_0\bx_1}$ for $\mu$-almost every $\bx_0,\bx_1$, i.e. when the conditional path laws coincide. It follows that
\begin{eqnarray}\label{minimizer}
\inf_{\mQ\in{\cal D}(\mu_0,\mu_1)}D(\mQ|\bP)=D(\mu^\star|\bar{\mu}),
\end{eqnarray}
with minimizer given by $\mQ^\star=\mQ^{^\star\bx_0\bx_1}\mu^\star(\bx_0,\bx_1)$, where $\mQ^{^\star\bx_0\bx_1}=\bP((\bx_t)_{t\in(0,1)}|\bx_0,\bx_1)$ and $\mu^\star$ is the solution to the static SB problem
\begin{eqnarray}\nn
\mu^\star=\inf_{\mu\in\Pi(\mu_0,\mu_1)}D(\mu|\bar{\mu}),
\end{eqnarray}
with $\Pi(\mu_0,\mu_1)$ denotes the set of all couplings between $\mu_0$ and $\mu_1$. In other words, the optimal $\mQ^\star$ is a $\mu^\star$-weighted mixture of reference trajectories. This representation often admits a closed-form expression and highlights how the target distribution $\mu^\star$ reweights the reference process.

According to Lemma \ref{lemma}, there exists a pair of functions $\hat{\rho}_0$ and $\rho_1$ on $\bR^d$ such that the measure $\mu^\star$ on $\bR^d\times\bR^d$ can be expressed as 
\begin{eqnarray}\label{paired}
\mu^\star(\bx_0,\bx_1)=q(0,\bx_0,1,\bx_1)\hat{\rho}_0(\bx_0)\rho_1(\bx_1)
\end{eqnarray}
which has marginals $\mu_0$ and $\mu_1$, i.e. 
\begin{eqnarray}\label{marginal}
\mu_0(\bx_0)=\rho_0(\bx_0)\hat{\rho}_0(\bx_0),&&\mu_1(\bx_1)=\rho_1(\bx_1)\hat{\rho}_1(\bx_1),
\end{eqnarray}
where 
\begin{eqnarray}\nn
\rho_0(\bx_0)=\int q(0,\bx_0,1,\bx_1)\rho_1(\bx_1)d\bx_1,&&\hat{\rho}_1(\bx_1)=\int q(0,\bx_0,1,\bx_1)\hat{\rho}_0(\bx_0)d\bx_0.
\end{eqnarray}

Recall that the reference path measure $\bP$ corresponds to the solution of the SDE
\begin{eqnarray}\label{sden}
d\bx_t=\bb(\bx_t,t)dt+\sigma(t)d\bw_t, \qquad \bx_0\sim\mu_0.
\end{eqnarray}
The H\"{o}lder condition imposed on $\bb(\bx,t)$ guarantees that the drift is globally Lipschitz, which ensures that this SDE has a unique strong solution \citep{10.5555/129416}. 

To construct the optimal $\mQ$, we posit a gradient ansatz: assume that $\mQ$ is induced by
\begin{eqnarray}\label{sde2}
d\bx_t=\big[\bb(\bx_t,t)-\sigma(t)^2\nabla \phi(\bx_t,t)\big]dt+\sigma(t)d\bw_t,\qquad \bx_0\sim\mu_0,
\end{eqnarray}
for some smooth potential $\phi:\bR^d\times[0,1]\to\bR$. The role of $\phi$ is to steer the paths so that the law of $\mQ$ matches $\mQ^\star$, the minimizer in (\ref{minimizer}).
Because (\ref{sden}) and (\ref{sde2}) share the same diffusion term, Girsanov’s theorem applies. The Radon–Nikodym derivative of $\mQ$ with respect to $\bP$ is
\begin{eqnarray}\label{rnd}
\frac{d\mQ}{d\bP}=\exp\left(-\tfrac{1}{2}\int_0^1\sigma(t)^2|\nabla\phi(\bx_t,t)|^2dt-\int_0^1\sigma(t)\nabla\phi(\bx_t,t)\cdot d\bw_t\right).
\end{eqnarray}

Applying Itô’s lemma to $\phi(\bx_t,t)$, we can rewrite the stochastic integral in (\ref{rnd}) and arrive at
\begin{eqnarray}\nn
d\phi(\bx_t,t)=\left(\partial_t\phi(\bx_t,t)+\nabla \phi(\bx_t,t)\cdot \bb(\bx_t,t)+\frac{\sigma(t)^2}{2}\bigtriangleup \phi(\bx_t,t)\right)dt+\sigma(t)\nabla \phi(\bx_t,t)\cdot d\bw_t.
\end{eqnarray}
where $\bigtriangleup$ is the Laplacian. Integrating and rearranging, we can express the stochastic integral as
\begin{eqnarray}\nn
&&-\int_0^1\sigma(t)\nabla \phi(\bx_t,t)\cdot d\bw_t\\\nn
&=&\int_0^1\left(\partial_t\phi(\bx_t,t)+\nabla \phi(\bx_t,t)\cdot\bb(\bx_t,t)+\frac{\sigma(t)^2}{2}\bigtriangleup \phi(\bx_t,t)\right)dt+\phi(\bx_0,0)-\phi(\bx_1,1).
\end{eqnarray} 
Substituting back into (\ref{rnd}) gives
\begin{eqnarray}\nn
\frac{d\mQ}{d\bP}&=&\exp\left\{\phi(\bx_0,0)-\phi(\bx_1,1)\right.\\\nn
&&\left.+\int_0^1\left(\partial_t\phi(\bx_t,t)+\nabla \phi(\bx_t,t)\cdot\bb(\bx_t,t)+\frac{\sigma(t)^2}{2}\bigtriangleup \phi(\bx_t,t)-\frac{\sigma(t)^2}{2}\|{\nabla} \phi(\bx_t,t)\|^2\right)dt\right\}.
\end{eqnarray}
Choosing $\phi$ to satisfy the PDE
\begin{eqnarray}\label{pde}
\partial_t\phi+\nabla \phi\cdot\bb+\tfrac{\sigma(t)^2}{2}\Delta\phi-\tfrac{\sigma(t)^2}{2}|\nabla \phi|^2=0
\end{eqnarray}
yields 
\begin{eqnarray}\label{derivative}
\frac{d\mQ}{d\bP}=\exp\{\phi(\bx_0,0)-\phi(\bx_1,1)\}. 
\end{eqnarray}
The PDE (\ref{pde}) is nonlinear because of the quadratic term $|\nabla \phi|^2$. To handle this difficulty, we apply the Cole–Hopf transformation \citep{hopf1950}, a classical device for linearizing such equations. Specifically, we set $\phi(\bx,t)=-\log h(\bx,t)$. Substituting this relation into (\ref{pde}) and simplifying, one finds that the nonlinear term cancels out, leaving the linear parabolic PDE
\begin{eqnarray}\label{lpde}
\partial_t h + \nabla h \cdot \bb + \tfrac{\sigma(t)^2}{2}\Delta h = 0,
\end{eqnarray}
with the boundary conditions $h(\bx_0,0)=\exp\{-\phi(\bx_0,0)\}$ and $h(\bx_1,1)=\exp\{-\phi(\bx_1,1)\}$. Thus the nonlinear PDE (\ref{pde}) with terminal conditions for $\phi$ is equivalently reformulated as the linear PDE (\ref{lpde}) with terminal conditions for $h$. 

Observe that (\ref{lpde}) coincides with the backward Fokker–Planck equation (\ref{backward}). Accordingly, we consider the solution $h(\bx,t)=\rho_t(\bx)$ subject to the boundary conditions $h(\bx_0,0)=\rho_0(\bx_0)$ and $h(\bx_1,1)=\rho_1(\bx_1)$. Substituting this expression into (\ref{derivative}) and applying (\ref{paired}) yields
\begin{eqnarray}\nn
\frac{d\mQ}{d\bP}
= \frac{\rho_1(\bx_1)}{\rho_0(\bx_0)}
= \frac{q(0,\bx_0,1,\bx_1)\hat{\rho}_0(\bx_0)\rho_1(\bx_1)}
{q(0,\bx_0,1,\bx_1)\hat{\rho}_0(\bx_0)\rho_0(\bx_0)}
= \frac{\mu^\star}{\bar{\mu}}(\bx_0,\bx_1),
\end{eqnarray}
from which it follows that $D(\mQ^\star\|\bP)=D(\mu^\star\|\bar{\mu})$. Consequently, the stochastic differential equation (\ref{sde2}) with drift correction determined by $\phi(\bx,t)=-\log\rho_t(\bx)$ realizes the optimal law $\mQ^\star$.

The partial differential equation (\ref{lpde}) can also be solved directly, as it corresponds to the backward Fokker–Planck equation associated with the stochastic differential equation
\begin{eqnarray}\nn
d\bx_t = \bb(\bx_t,t),dt + \sigma(t),d\bw_t,
\qquad \bx_1 \sim \rho_1.
\end{eqnarray}
Applying Itô’s lemma to $h(\bx_t,t)$ yields
\begin{eqnarray}\nn
dh(\bx_t,t)&=&\partial_th(\bx_t,t)+\nabla h(\bx_t,t)\cdot\bb(\bx_t,t)+\frac{\sigma(t)^2}{2}\bigtriangleup h(\bx_t,t)+\sigma(t)\nabla h(\bx_t,t)\cdot d\bw_t\\\nn
&=&\sigma(t)\nabla h(\bx_t,t)\cdot d\bw_t.
\end{eqnarray}	 
where the last equality follows from the fact that $h$ satisfies (\ref{lpde}). Integrating from $t$ to 1 gives
\begin{eqnarray}\nn
h(\bx_1,1)-h(\bx_t,t)&=&\int_t^1\sigma(t)\nabla h(\bx_t,t)\cdot d\bw_t.
\end{eqnarray}	 
Taking conditional expectations under $\bP$ leads to the Feynman–Kac representation: 
\begin{eqnarray}\label{solution}
h(\bx,t)=\mE_{\bP}\left[\rho_1(\bx_1)|\bx_t=\bx\right]=\int q(t,\bx,1,\bx_1)\rho_1(\bx_1)d\bx_1=\rho_t(\bx).
\end{eqnarray}	 
The boundary condition is then verified as   
\begin{eqnarray}\nn
h(\bx_0,0)=\mE_{\bP}\left[\rho_1(\bx_1)|\bx_0=\bx_0\right]=\int q(0,\bx_0,1,\bx_1)\rho_1(\bx_1)d\bx_1=\rho_0(\bx_0),
\end{eqnarray}	 
confirming consistency. Transforming back to $\phi$ yields the corresponding drift term in (\ref{sde2})
\begin{eqnarray}\nn
\bu^\star(\bx,t)&=&\sigma(t)^2\nabla\log\rho_t(\bx)=\frac{\sigma(t)^2\int\nabla_{\bx}q(t,\bx,1,\bx_1)\rho_1(\bx_1)d\bx_1}{\rho_t(\bx)}\\\nn
&=&\frac{\sigma(t)^2\int\nabla_{\bx}q(t,\bx,1,\bx_1)\rho_1(\bx_1)d\bx_1\hat{\rho}_t(\bx)}{\rho_t(\bx)\hat{\rho}_t(\bx)}\\\nn
&=&\frac{\sigma(t)^2\int\nabla_{\bx}q(t,\bx,1,\bx_1)\rho_1(\bx_1)d\bx_1\int q(0,\bx_0,t,\bx)\hat{\rho}_0(\bx_0)d\bx_0}{\int q(t,\bx,1,\bx_1)\rho_1(\bx_1)d\bx_1\int q(0,\bx_0,t,\bx)\hat{\rho}_0(\bx_0)d\bx_0}\\\nn
&=&\frac{\sigma(t)^2\int\int\{\nabla_{\bx}q(t,\bx,1,\bx_1)\}q(0,\bx_0,t,\bx)\frac{\mu(\bx_0,\bx_1)}{q(0,\bx_0,1,\bx_1)}d\bx_0d\bx_1}{\int\int q(t,\bx,1,\bx_1)q(0,\bx_0,t,\bx)\frac{\mu(\bx_0,\bx_1)}{q(0,\bx_0,1,\bx_1)}d\bx_0d\bx_1}\\\label{finalu}
&=&\frac{\sigma(t)^2\int\int s_{\bx_1}(\bx,t)g_t(\bx,\bx_0,\bx_1)\mu(\bx_0,\bx_1)d\bx_0d\bx_1}{\int\int g_t(\bx,\bx_0,\bx_1)\mu(\bx_0,\bx_1)d\bx_0d\bx_1},
\end{eqnarray}	 
which coincides with the expression in (\ref{drift0}). Here, $\hat{\rho}_t(\bx)$ is defined in (\ref{rhot1}), (\ref{paired}) has been applied, and $g_t$ enotes the transition kernel defined in (\ref{transition}).

\section{Proof of Corollary \ref{coro}}\label{Gaussian}
We begin with the well-known closed-form solution of a linear stochastic differential equation (SDE). Consider the reference SDE
\begin{eqnarray}\label{linear}
d\bx_t = {c(t)\bx_t + \alpha(t)}dt + \sigma(t),d\bw_t,
\end{eqnarray}
which admits an explicit solution. Define the auxiliary functions
\begin{eqnarray}\nn
\tau(t) = \exp\left(\int_0^t c(s)ds\right), \quad
\tau_1(t) = \exp\left(\int_t^1 c(s)ds\right) = \frac{\tau(1)}{\tau(t)},\\\nn
\kappa(t,t) = \tau(t)^2 \int_0^t \frac{\sigma(s)^2}{\tau(s)^2}ds, \quad
\kappa_1(t,t) = \tau(1)^2 \int_t^1 \frac{\sigma(s)^2}{\tau(s)^2}ds,\\\nn
\zeta(t) = \tau(t)\int_0^t \frac{\alpha(s)}{\tau(s)}ds, \quad
\zeta_1(t) = \tau(1)\int_t^1 \frac{\alpha(s)}{\tau(s)}ds.
\end{eqnarray}
Note that these quantities satisfy the useful relations
\begin{eqnarray}\nn
\frac{\zeta_1(t)}{\tau(1)} = \frac{\zeta(1)}{\tau(1)} - \frac{\zeta(t)}{\tau(t)}, \quad
\frac{\kappa(1,1)}{\tau(1)^2} - \frac{\kappa(t,t)}{\tau(t)^2} = \frac{\kappa_1(t,t)}{\tau(1)^2}.
\end{eqnarray}
The solution to (\ref{linear}) is then given by \citep{PlatenEckhard2010NSoS}
\begin{eqnarray}\nn
\bx_t = \tau(t)\left(\bx_0 + \int_0^t \frac{\alpha(s)}{\tau(s)}ds + \int_0^t \frac{\sigma(s)}{\tau(s)}d\bw_s\right).
\end{eqnarray}
Consequently, the transition kernel of the process takes the Gaussian form
\begin{eqnarray}\nn
q(0,\bx_0,t,\bx) &=& N(\bx;,\eta(t),\kappa(t,t)\bI_d),\\\label{gkernel}
q(0,\bx_0,1,\bx_1) &=& N(\bx_1;,\eta(1),\kappa(1,1)\bI_d),\\\nn
q(t,\bx,1,\bx_1) &=& N(\bx_1;,\eta_1(t),\kappa_1(t,t)\bI_d),
\end{eqnarray}
where the conditional means are
\begin{eqnarray}\nn
\eta(t) = \mE[\bx_t \mid \bx_0] = \tau(t)\bx_0 + \zeta(t), \quad
\eta_1(t) = \mE[\bx_1 \mid \bx_t = \bx] = \tau_1(t)\bx + \zeta_1(t),
\end{eqnarray}
and the conditional covariances are
\begin{eqnarray}\nn
\mE[(\bx_t - \eta(t))(\bx_t - \eta(t))^\top \mid \bx_0] = \kappa(t,t)\bI_d, \quad
\mE[(\bx_1 - \eta_1(t))(\bx_1 - \eta_1(t))^\top \mid \bx_t = \bx] = \kappa_1(t,t)\bI_d.
\end{eqnarray}
Finally, the conditional score function is given by
\begin{eqnarray}\label{cscore}
s_{\bx_1}(\bx,t) = \nabla_{\bx}\log q(t,\bx,1,\bx_1)
= \frac{\tau_1(t)\{\bx_1 - \tau_1(t)\bx - \zeta_1(t)}\}{\kappa_1(t,t)}.
\end{eqnarray}

Next, we derive the solution to the static Schrödinger Bridge (SB) problem (\ref{ot}) between two given Gaussian distributions.
When the drift and time-dependent coefficients vanish, i.e., $c(t)=0,\alpha(t)=0,\sigma(t)=\sigma$, the dynamic problem reduces to the static entropy-regularized optimal transport (OT) problem
\begin{eqnarray}\label{ot0}
\mu^\star(\bx_0,\bx_1)=\text{argmin}_{\mu\in\Pi(\mu_0,\mu_1)}\left\{\int\|\bx_1-\bx_0\|^2\mu(d\bx_1,d\bx_1)+2\sigma^2\log\mu d\mu\right\}.
\end{eqnarray}
Here, the first term represents the quadratic transport cost, while the second term corresponds to the entropy regularization controlled by the diffusion scale $\sigma^2$.
Let the marginals be Gaussian, $\mu_0=N(\nu_0,\bSigma_0)$ and $\mu_1=N(\nu_1,\bSigma_1)$. For this setting, \citet{mallasto2020entropyregularized2wassersteindistancegaussian} provided a closed-form analytical expression for the optimal coupling.
Define
\begin{eqnarray}\label{dc}
D_\sigma = \big(4\bSigma_0^{1/2}\bSigma_1\bSigma_0^{1/2} + \sigma^4\bI_d\big)^{1/2},
\qquad
C_\sigma = \tfrac{1}{2}\big(\bSigma_0^{1/2} D_\sigma \bSigma_0^{-1/2} - \sigma^2\bI_d\big),
\end{eqnarray}
then the optimal coupling $\mu^\star$ to (\ref{ot0}) is itself a Gaussian:
\begin{eqnarray}\label{closed}
\mu^\star\sim N\left(\left(\begin{array}{c}\nu_0\\\nu_1\end{array}\right),\left(\begin{array}{cc}\Sigma_0&C_\sigma\\C_\sigma^T&\bSigma_1\end{array}\right)\right).
\end{eqnarray}

We now extend this result to the general linear drift case. Substituting the Gaussian transition kernel (\ref{gkernel}) into (\ref{ot}) gives
\begin{eqnarray}\label{ot2}
\min_{\tilde{\mu}\in\Pi(\tilde{\mu}_0,\mu_1)}\left\{\int\|\bx_1-(\tau(1)\bx_0+\zeta(1))\|^2\mu(d\bx_0,d\bx_1)+2\kappa(1,1)\int\log\mu d\mu\right\}.
\end{eqnarray}
The presence of the linear drift modifies both the transport cost (through the rescaling factor $\tau(1)$ and shift $\zeta(1)$) and the entropy regularization (through the effective variance term $\kappa(1,1)$. To simplify the problem, consider the change of variable $\tilde{\bx}_0=\tau(1)\bx_0+\zeta(1)$. This transformation rescales and shifts the initial distribution, producing a new joint law $\tilde{\mu}$ with marginals $\tilde{\mu}_0\sim N(\tilde{\nu}_0,\tilde{\bSigma}_0))$ and $\mu_1=N(\nu_1,\bSigma_1)$, where $\tilde{\nu}_0=\tau(1)\nu_0+\zeta(1)$, $\tilde{\bSigma}_0=\tau(1)^2\bSigma_0$. Since the transformation is bijective, minimizing over $\mu$ or $\tilde{\mu}$ is equivalent. The first term in (\ref{ot2}) then becomes the standard quadratic cost $\mE\|\bx_1-\tilde{\bx}_0\|^2$,
while the entropy term is invariant up to an additive constant, $\int\log\tilde{\mu}d\tilde{\mu}=\int\log\mu d\mu+\text{const}$. Hence, (\ref{ot2}) is equivalent to
\begin{eqnarray}\label{ot2}
\min_{\tilde{\mu} \in \Pi(\tilde{\mu}_0, \mu_1)}
\left\{
\int \|\bx_1 -  \tilde{\bx}_0\|^2 \tilde{\mu}(d\bx_0,d\bx_1)+2\kappa(1,1) \int \log \tilde{\mu}, d\tilde{\mu}
\right\}.
\end{eqnarray}
Equation (\ref{ot2}) has the same form as (\ref{ot0}), but with the adjusted variance parameter $\tilde{\sigma}^2=\kappa(1,1)$. Therefore, the optimal coupling $\tilde{\mu}^\star$ is again Gaussian:
\begin{eqnarray}\label{closed1}
\tilde{\mu}^\star\sim N\left(\left(\begin{array}{c}\tilde{\nu}_0\\\nu_1\end{array}\right),\left(\begin{array}{cc}\tilde{\Sigma}_0&\tilde{C}_{\tilde{\sigma}}\\\tilde{C}_{\tilde{\sigma}}^T&\bSigma_1\end{array}\right)\right),
\end{eqnarray}
where  
\begin{eqnarray}\label{dc1}
\tilde{D}_{\tilde{\sigma}}=(4\tilde{\bSigma}_0^{1/2}\bSigma_1\tilde{\bSigma}_0^{1/2}+\tilde{\sigma}^4\bI_d)^{1/2},&&\tilde{C}_{\tilde{\sigma}}=\frac{1}{2}(\tilde{\bSigma}_0^{1/2}\tilde{D}_{\tilde{\sigma}}\tilde{\bSigma}_0^{-1/2}-{\tilde{\sigma}}^2\bI_d).
\end{eqnarray}
Finally, applying the inverse transformation $\bx_0=\tau(1)^{-1}\tilde{\bx}_0-\zeta(1)$ maps the optimal coupling back to the original coordinates.
The resulting optimal static Gaussian SB takes the form
\begin{eqnarray}\label{closed1}
\mu^\star\sim N\left(\left(\begin{array}{c}\nu_0\\\nu_1\end{array}\right),\left(\begin{array}{cc}\bSigma_0&C_{\sigma_\star}\\C_{\sigma_\star}^T&\bSigma_1\end{array}\right)\right),
\end{eqnarray}
where the effective regularization scale is $\sigma_\star^2=\kappa(1,1)/\tau(1)$. 

To evaluate (\ref{finalu}), we next integrate over both $\bx_0$ and $\bx_1$. Since $\mu^\star(\bx_0,\bx_1)$ in (\ref{closed1}) is multivariate normal, we can collect all quadratic and linear terms in $\bx_0$ and $\bx_1$, express the integrand as an exponential of a joint quadratic form, and then perform the Gaussian integration analytically.
\begin{eqnarray}\nn
&&[(\bx_0-{\nu}_{0})^T,(\bx_1-{\nu}_{1})^T]\left(\begin{array}{cc}\bSigma_0&C_{\sigma^\star}\\C_{\sigma^\star}^T&\bSigma_1\end{array}\right)^{-1}\left(\begin{array}{c}\bx_0-{\nu}_{0}\\\bx_1-{\nu}_{1}\end{array}\right)\\\nn
&&+\frac{\{\bx-\tau(t)\bx_0-\zeta(t)\}^T\{\bx-\tau(t)\bx_0-\zeta(t)\}}{\kappa(t,t)}\\\nn
&&+\frac{\{\bx_1-\tau_1(t)\bx-\zeta_1(t))\}^T\{\bx_1-\tau_1(t)\bx-\zeta_1(t))\}}{\kappa_1(t,t)}\\
&&-\frac{\{\bx_1-\tau(1)\bx_0-\zeta(1)\}^T\{\bx_1-\tau(1)\bx_0-\zeta(1)\}}{\kappa(1,1)}\\\nn
&=&[(\bx_0-{\nu}_{0})^T,(\bx_1-{\nu}_{1})^T]\left(\begin{array}{cc}\bSigma_0&C_{\sigma^\star}\\C_{\sigma^\star}^T&\bSigma_1\end{array}\right)^{-1}\left(\begin{array}{c}\bx_0-{\nu}_{0}\\\bx_1-{\nu}_{1}\end{array}\right)\\\nn
&&+(\bx_0^T,\bx_1^T)\left(\begin{array}{cc}\left\{\frac{\tau(t)^2}{\kappa(t,t)}-\frac{\tau(1)^2}{\kappa(1,1)}\right\}\bI_d&\frac{\tau(1)}{\kappa(1,1)}\bI_d\\\frac{\tau(1)}{\kappa(1,1)}\bI_d&\left\{\frac{1}{\kappa_1(t,t)}-\frac{1}{\kappa(1,1)}\right\}\bI_d\end{array}\right)\left(\begin{array}{c}\bx_0\\\bx_1\end{array}\right)\\\nn
&&-2\left[\frac{\tau(t)(\bx-\zeta(t))^T}{\kappa(t,t)}+\frac{\tau(1)\zeta(1)}{\kappa(1,1)},\frac{(\tau_1(t)\bx+\zeta_1(t))^T}{\kappa_1(t,t)}-\frac{\zeta(1)}{\kappa(1,1)}\right]\left(\begin{array}{c}\bx_0\\\bx_1\end{array}\right)+\text{const}.
\end{eqnarray}
Therefore, the joint density of $\left(\begin{array}{c}\bx_0\\\bx_1\end{array}\right)$ is a multivariate normal distribution in $\bR^{2d}$, with mean vector given by
\begin{eqnarray}\nn
\left(\begin{array}{c}\bar{\bx}_0\\\bar{\bx}_1\end{array}\right)&=&\left\{\left(\begin{array}{cc}\frac{\tau(t)^2\kappa_1(t,t)}{\kappa(t,t)\kappa(1,1)}\bI_d&\frac{\tau(1)}{\kappa(1,1)}\bI_d\\\frac{\tau(1)}{\kappa(1,1)}\bI_d&\frac{\tau(1)^2\kappa(t,t)}{\tau(t)^2\kappa_1(t,t)\kappa(1,1)}\bI_d\end{array}\right)+\left(\begin{array}{cc}\bSigma_0&C_{\sigma^\star}\\C_{\sigma^\star}^T&\bSigma_1\end{array}\right)^{-1}\right\}^{-1}\\\nn
&&\left[\left(\begin{array}{cc}\bSigma_0&C_{\sigma^\star}\\C_{\sigma^\star}^T&\bSigma_1\end{array}\right)^{-1}\left(\begin{array}{c}{\nu}_{0}\\{\nu}_{1}\end{array}\right)+\left(\begin{array}{c}\frac{\tau(t)(\bx-\zeta(t))}{\kappa(t,t)}+\frac{\tau(1)\zeta(1)}{\kappa(1,1)}\\\frac{\tau_1(t)\bx+\zeta_1(t)}{\kappa_1(t,t)}-\frac{\zeta(1)}{\kappa(1,1)}\end{array}\right)\right]\\\nn
&=&\left\{\bI_d-A^{-1}\bu(\bI_d+\bv^TA^{-1}\bu)^{-1}\bv^T\right\}\\\label{jointm}
&&\left[\left(\begin{array}{c}{\nu}_{0}\\{\nu}_{1}\end{array}\right)+A^{-1}\bu\otimes\left\{\bx-\zeta(t)+\frac{\tau(1)\kappa(t,t)}{\tau(t)\kappa(1,1)}\zeta(1)\right\}\right],
\end{eqnarray}
where 
\begin{eqnarray}\nn
\bu=\left(\begin{array}{c}\frac{\tau(t)}{\kappa(t,t)}\\\frac{\tau(1)}{\tau(t)\kappa_1(t,t)}\end{array}\right),~~
\bv=\left(\begin{array}{c}\frac{\tau(t)\kappa_1(t,t)}{\kappa(1,1)}\\\frac{\tau(1)\kappa(t,t)}{\tau(t)\kappa(1,1)}\end{array}\right),~~
A=\left(\begin{array}{cc}\bSigma_0&C_{\sigma^\star}\\C_{\sigma^\star}^T&\bSigma_1\end{array}\right)^{-1},
\end{eqnarray}
and we have used the matrix inversion identity stating that if $A$ is invertible and $\bu,\bv$ re column vectors satisfying $1+\bv^TA^{-1}\bu\ne 0$, then
\begin{eqnarray}\nn
(A+\bu\bv^T)^{-1}=A^{-1}-A^{-1}\bu(1+\bv^TA^{-1}\bu)^{-1}\bv^T A^{-1}.
\end{eqnarray}
Define
\begin{eqnarray}\nn
\bSigma_t=\bI_d+\bv^TA^{-1}\bu=\bI_d+\frac{\tau(t)^2\kappa_1(t,t)}{\kappa(t,t)\kappa(1,1)}\bSigma_0+\frac{\tau(1)^2\kappa(t,t)}{\tau(t)^2\kappa_1(t,t)\kappa(1,1)}\bSigma_1+\frac{\tau(1)}{\kappa(1,1)}(C_{\sigma^\star}+C_{\sigma^\star}^T),
\end{eqnarray}
which differs slightly—up to a constant multiplicative factor—from the analogous notation used in \citet{bunne2023schrodingerbridgegaussianmeasures}, namely,
\begin{eqnarray}\nn
&&\frac{\kappa(t,t)\kappa_1(t,t)}{\kappa(1,1)}\bI_d+\frac{\tau(t)^2\kappa_1(t,t)^2}{\kappa(1,1)^2}\bSigma_0+\frac{\tau(1)^2\kappa(t,t)^2}{\tau(t)^2\kappa(1,1)^2}\bSigma_1+\frac{\tau(1)\kappa(t,t)\kappa_1(t,t)}{\kappa(1,1)^2}(C_{\sigma^\star}+C_{\sigma^\star}^T)\\\nn
&=&\kappa(t,t)(1-\rho(t))\bI_d+\bar{r}(t)^2\bSigma_0+r(t)^2\bSigma_1+r(t)\bar{r}(t)(C_{\sigma^\star}+C_{\sigma^\star}^T).
\end{eqnarray}
We can write (\ref{jointm}) as
\begin{eqnarray}\nn
\left(\begin{array}{c}\bar{\bx}_0\\\bar{\bx}_1\end{array}\right)&=&\left(\begin{array}{c}{\bnu}_{0}\\{\bnu}_{1}\end{array}\right)-\left(\begin{array}{cc}\bSigma_0&C_{\sigma^\star}\\C_{\sigma^\star}^T&\bSigma_1\end{array}\right)\left(\begin{array}{cc}\frac{\tau(t)^2\kappa_1(t,t)}{\kappa(t,t)\kappa(1,1)}\bI_d&\frac{\tau(1)}{\kappa(1,1)}\bI_d\\\frac{\tau(1)}{\kappa(1,1)}\bI_d&\frac{\tau(1)^2\kappa(t,t)}{\tau(t)^2\kappa_1(t,t)\kappa(1,1)}\bI_d\end{array}\right)\bSigma_t^{-1}\left(\begin{array}{c}{\bnu}_{0}\\{\bnu}_{1}\end{array}\right)\\\nn,
&&+\left(\begin{array}{cc}\bSigma_0&C_{\sigma^\star}\\C_{\sigma^\star}^T&\bSigma_1\end{array}\right)\left(\begin{array}{c}\frac{\tau(t)}{\kappa(t,t)}\bI_d\\\frac{\tau(1)}{\tau(t)\kappa_1(t,t)}\bI_d\end{array}\right)\bSigma_t^{-1}\left\{\bx-\zeta(t)+\frac{\tau(1)\kappa(t,t)}{\tau(t)\kappa(1,1)}\zeta(1)\right\}.
\end{eqnarray}

Further denote
\begin{eqnarray}\nn
r(t)=\frac{\tau(1)\kappa(t,t)}{\tau(t)\kappa(1,1)}&&\bar{r}(t)=\frac{\tau(t)\kappa_1(t,t)}{\kappa(1,1)},\\\nn
\dot{r}(t)=c(t)r(t)+\frac{\sigma(t)^2\tau(1)}{\tau(t)\kappa(1,1)}&&\dot{\bar{r}}(t)=c(t)\bar{r}(t)-\frac{\sigma(t)^2\tau(1)^2}{\tau(t)\kappa(1,1)},\\\nn
\dot{\zeta}(t)=c(t)\zeta(t)+\alpha(t)&&
\end{eqnarray}
In (\ref{finalu}), we only need to know $\bar{\bx}_1$ which is 
\begin{eqnarray}\nn
\bar{\bx}_1&=&\bnu_{1}-\left(\begin{array}{cc}\frac{\tau(t)^2\kappa_1(t,t)}{\kappa(t,t)\kappa(1,1)}C_{\sigma^\star}^T+\frac{\tau(1)}{\kappa(1,1)}\bSigma_1&\frac{\tau(1)}{\kappa(1,1)}C_{\sigma^\star}^T+\frac{\tau(1)^2\kappa(t,t)}{\tau(t)^2\kappa_1(t,t)\kappa(1,1)}\bSigma_1\end{array}\right)\bSigma_t^{-1}\left(\begin{array}{c}{\bnu}_{0}\\{\bnu}_{1}\end{array}\right)\\\nn,
&&+\left(\frac{\tau(t)}{\kappa(t,t)}C_{\sigma^\star}^T+\frac{\tau(1)}{\tau(t)\kappa_1(t,t)}\bSigma_1\right)\bSigma_t^{-1}\left\{\bx-\zeta(t)+\frac{\tau(1)\kappa(t,t)}{\tau(t)\kappa(1,1)}\zeta(1)\right\}\\\nn
&=&\bnu_{1}-\left(\frac{\tau(t)^2\kappa_1(t,t)}{\kappa(t,t)\kappa(1,1)}C_{\sigma^\star}^T-\frac{\tau(1)}{\kappa(1,1)}\bSigma_1\right)\bSigma_t^{-1}\bnu_{0}+\left(\frac{\tau(1)}{\kappa(1,1)}C_{\sigma^\star}^T+\frac{\tau(1)^2\kappa(t,t)}{\tau(t)^2\kappa_1(t,t)\kappa(1,1)}\bSigma_1\right)\bSigma_t^{-1}\bnu_{1}\\\nn,
&&+\left(\frac{\tau(t)}{\kappa(t,t)}C_{\sigma^\star}^T+\frac{\tau(1)}{\tau(t)\kappa_1(t,t)}\bSigma_1\right)\bSigma_t^{-1}\left\{\bx-\zeta(t)+\frac{\tau(1)\kappa(t,t)}{\tau(t)\kappa(1,1)}\zeta(1)\right\}\\\nn
&=&\bnu_{1}-\left(\bar{r}(t)^2C_{\sigma^\star}^T+r(t)\bar{r}(t)\bSigma_1\right)\bSigma_t^{-1}\bnu_{0}-\left(r(t)\bar{r}(t)C_{\sigma^\star}^T+r(t)^2\bSigma_1\right)\bSigma_t^{-1}\bnu_{1}\\\nn,
&&+\left(\bar{r}(t)C_{\sigma^\star}^T+r(t)\bSigma_1\right)\bSigma_t^{-1}\left\{\bx-\zeta(t)+r(t)\zeta(1)\right\}.
\end{eqnarray}
Substitute into (\ref{cscore}) and (\ref{finalu}), we have
\begin{eqnarray}\nn
\bu^\star(\bx,t)&=&\sigma(t)^2\frac{\tau_1(t)\{\bar{\bx}_1-\tau_1(t)\bx-\zeta_1(t))\}}{\kappa_1(t,t)}\\\nn
&=&\sigma(t)^2\frac{\tau_1(t)}{\kappa_1(t,t)}\{\bar{\bx}_1-\tau_1(t)[\bx-\zeta(t)]-\zeta(1)\}.
\end{eqnarray}
According to (\ref{drift0}), the final drift is
\begin{eqnarray}\nn
&&c(t)\bx+\alpha(t)+\sigma(t)^2\frac{\tau_1(t)}{\kappa_1(t,t)}\{\bar{\bx}_1-\tau_1(t)[\bx-\zeta(t)]-\zeta(1)\}\\\nn
&=&c(t)\bx+\alpha(t)+\sigma(t)^2\frac{\tau_1(t)}{\kappa_1(t,t)}[\left(r(t)\bSigma_1+\bar{r}(t)C_{\sigma^\star}^T\right)\bSigma_t^{-1}\left\{\bx-\zeta(t)+r(t)\zeta(1)\right\}\\\nn
&&+\bnu_{1}-\left(\bar{r}(t)^2C_{\sigma^\star}^T+r(t)\bar{r}(t)\bSigma_1\right)\bSigma_t^{-1}\bnu_{0}-\left(r(t)\bar{r}(t)C_{\sigma^\star}^T+r(t)^2\bSigma_1\right)\bSigma_t^{-1}\bnu_{1}\\\nn
&&-\tau_1(t)[\bx-\zeta(t)]-\zeta(1)]\\\nn
&=&c(t)\bx+\alpha(t)+\sigma(t)^2\frac{\tau_1(t)}{\kappa_1(t,t)}[\left(r(t)\bSigma_1+\bar{r}(t)C_{\sigma^\star}^T\right)\bSigma_t^{-1}\\\nn
&&\left\{\bx-\bar{r}(t)\bnu_0-r(t)\bnu_1-\zeta(t)+r(t)\zeta(1)\right\}+\bnu_{1}-\tau_1(t)[\bx-\zeta(t)]-\zeta(1)]\\\nn
&=&c(t)\bx+\alpha(t)+\sigma(t)^2\frac{\tau_1(t)}{\kappa_1(t,t)}[\left(r(t)\bSigma_1+\bar{r}(t)C_{\sigma^\star}^T\right)\bSigma_t^{-1}-\tau_1(t)\bI_d]\\\nn
&&\left\{\bx-\bar{r}(t)\bnu_0-r(t)\bnu_1-\zeta(t)+r(t)\zeta(1)\right\}+\sigma(t)^2\frac{\tau_1(t)}{\kappa_1(t,t)}[(\bnu_{1}-\zeta(1))(1-\tau_1(t)r(t))-\tau_1(t)\bar{r}(t)\bnu_0]\\\nn
&=&\sigma(t)^2\left(\frac{\tau(1)}{\tau(t)\kappa(1,1)}[r(t)\bSigma_1+\bar{r}(t)C_{\sigma^\star}^T]-\frac{\tau(1)^2}{\tau(t)\kappa(1,1)}[\bar{r}(t)\bSigma_0+r(t)C_{\sigma^\star}]-\rho(t)\bI_d\right)\\\nn
&&\bSigma_t^{-1}\left\{\bx-\bar{r}(t)\bnu_0-r(t)\bnu_1-\zeta(t)+r(t)\zeta(1)\right\}\\\label{normaldrift}
&&+\frac{\sigma(t)^2\tau(1)}{\tau(t)\kappa(1,1)}(\bnu_1-\zeta(1))-\frac{\sigma(t)^2\tau(1)^2}{\tau(t)\kappa(1,1)}\bnu_0+c(t)\bx+\alpha(t).
\end{eqnarray}
To facilitate comparison with the well-known result in \citet{bunne2023schrodingerbridgegaussianmeasures}, we define
\begin{eqnarray}\nn
S(t)&=&\dot{r}(t)r(t)\bSigma_1+\dot{r}(t)\bar{r}(t)C_{\sigma^\star}+\dot{\bar{r}}(t)[\bar{r}(t)\bSigma_0+r(t)C_{\sigma^\star}^T]+[c(t)\kappa(t,t)(1-\rho(t))-\sigma(t)^2\rho(t)]\bI_d\\\nn
&=&c(t)\bSigma_t+\sigma(t)^2\left(\frac{\tau(1)}{\tau(t)\kappa(1,1)}[r(t)\bSigma_1+\bar{r}(t)C_{\sigma^\star}]-\frac{\tau(1)^2}{\tau(t)\kappa(1,1)}[\bar{r}(t)\bSigma_0+r(t)C_{\sigma^\star}^T]-\rho(t)\bI_d\right).
\end{eqnarray}
With this notation, equation (29) in Theorem 3 of \citet{bunne2023schrodingerbridgegaussianmeasures} can be equivalently written as
\begin{eqnarray}\nn
&&S(t)^T\bSigma_t^{-1}\left\{\bx-\bar{r}(t)\bnu_0-r(t)\bnu_1-\zeta(t)+r(t)\zeta(1)\right\}+\dot{\bar{r}}(t)\bnu_0+\dot{r}(t)\bnu_1+c(t)\zeta(t)+\alpha(t)-\dot{r}(t)\zeta(1)\\\nn
&=&\sigma(t)^2\left(\frac{\tau(1)}{\tau(t)\kappa(1,1)}[r(t)\bSigma_1+\bar{r}(t)C_{\sigma^\star}^T]-\frac{\tau(1)^2}{\tau(t)\kappa(1,1)}[\bar{r}(t)\bSigma_0+r(t)C_{\sigma^\star}]-\rho(t)\bI_d\right)\\\nn
&&\bSigma_t^{-1}\left\{\bx-\bar{r}(t)\bnu_0-r(t)\bnu_1-\zeta(t)+r(t)\zeta(1)\right\}\\\nn
&&+\frac{\sigma(t)^2\tau(1)}{\tau(t)\kappa(1,1)}(\bnu_1-\zeta(1))-\frac{\sigma(t)^2\tau(1)^2}{\tau(t)\kappa(1,1)}\bnu_0+c(t)\bx+\alpha(t),
\end{eqnarray}
which coincides exactly with (\ref{normaldrift}). It is worth noting that Theorem 3 in \citet{bunne2023schrodingerbridgegaussianmeasures} establishes this result using more advanced mathematical machinery—specifically, the central identity from quantum field theory and the infinitesimal generator framework. In contrast, our derivation arrives at the same expression through direct and elementary manipulations of the multivariate normal density, offering a more transparent and computationally accessible perspective.

\section{Proof of Theorem \ref{thm1}}
For the SDE in (\ref{sde}) with reference process given by (\ref{sde0}), we first note that
\begin{eqnarray}\nn
\bu^\star(\bx,t) = \mE_{\bQ_{\bx_0,\bx_1|\bx_t=\bx}}\left[\sigma(t)^2\bs_{\bx_1}(t,\bx_t)\right],
\end{eqnarray}
as shown in (\ref{ualt}). Now, for any vector-valued function $\bu(\bx, t): \mathbb{R}^d \times [0,1] \to \mathbb{R}^d$, consider the following risk functional:
\begin{eqnarray}\nn
&&\mE_{\bQ_{\bx_0,\bx_1,\bx_t}}\left[\left|\sigma(t)^2\bs_{\bx_1}(t,\bx_t)-\bu(\bx_t,t)\right|^2\right]\\\nn
&=&\mE_{\bQ_{\bx_0,\bx_1,\bx_t}}\left[\left|\sigma(t)^2\bs_{\bx_1}(t,\bx_t)-\bu^\star(\bx_t,t)+\bu^\star(\bx_t,t)-\bu(\bx_t,t)\right|^2\right]\\\nn
&=&\mE_{\bQ_{\bx_0,\bx_1,\bx_t}}\left[\left|\sigma(t)^2\bs_{\bx_1}(t,\bx_t)-\bu^\star(\bx_t,t)\right|^2\right]
+\mE_{\bQ_{\bx_0,\bx_1,\bx_t}}\left[\left|\bu^\star(\bx_t,t)-\bu(\bx_t,t)\right|^2\right]\\\nn
&&+2,\mE_{\bQ_{\bx_0,\bx_1,\bx_t}}\left[\left\langle\sigma(t)^2\bs_{\bx_1}(t,\bx_t)-\bu^\star(\bx_t,t),,\bu^\star(\bx_t,t)-\bu(\bx_t,t)\right\rangle\right]\\\label{ineq}
&\ge&\mE_{\bQ_{\bx_0,\bx_1,\bx_t}}\left[\left|\sigma(t)^2\bs_{\bx_1}(t,\bx_t)-\bu^\star(\bx_t,t)\right|^2\right].
\end{eqnarray}
To justify the inequality in (\ref{ineq}), note that the inner product term vanishes:
\begin{eqnarray}\nn
&&\mE_{\bQ_{\bx_0,\bx_1,\bx_t}}\left[\left\langle\sigma(t)^2\bs_{\bx_1}(t,\bx_t)-\bu^\star(\bx_t,t),\bu^\star(\bx_t,t)-\bu(\bx_t,t)\right\rangle\right]\\\nn
&=&\mE_{\bQ_{\bx_t}}\left[\mE_{\bQ_{\bx_0,\bx_1|\bx_t}}\left\langle\sigma(t)^2\bs_{\bx_1}(t,\bx_t)-\bu^\star(\bx_t,t),\bu^\star(\bx_t,t)-\bu(\bx_t,t)\right\rangle\right]\\\nn
&=&\mE_{\bQ_{\bx_t}}\left[\left\langle\mE_{\bQ_{\bx_0,\bx_1|\bx_t}}\left[\sigma(t)^2\bs_{\bx_1}(t,\bx_t)-\bu^\star(\bx_t,t)\right],\bu^\star(\bx_t,t)-\bu(\bx_t,t)\right\rangle\right]\\\nn
&=&\mE_{\bQ_{\bx_t}}\left[\left\langle\bu^\star(\bx_t,t)-\bu^\star(\bx_t,t),\bu^\star(\bx_t,t)-\bu(\bx_t,t)\right\rangle\right]=0,
\end{eqnarray}
where we used the identity $\sigma(t)^2\mathbb{E}{\bQ{\bx_0,\bx_1 | \bx_t}} [\bs_{\bx_1}(t, \bx_t)] = \bu^\star(\bx_t, t)$ from (\ref{ualt}).
Therefore, it follows from (\ref{ineq}) that
\begin{eqnarray}\nn
\mE_{\bQ_{\bx_0,\bx_1,\bx_t}}\left[\left|\sigma(t)^2\bs_{\bx_1}(t,\bx_t)-\bu(\bx_t,t)\right|^2\right]
\ge
\mE_{\bQ_{\bx_0,\bx_1,\bx_t}}\left[\left|\sigma(t)^2\bs_{\bx_1}(t,\bx_t)-\bu^\star(\bx_t,t)\right|^2\right],
\end{eqnarray}
with equality if and only if $\bu(\bx_t, t) = \bu^\star(\bx_t, t)$.
Hence, $\bu^\star(\bx_t, t)$ uniquely minimizes the objective in (\ref{lse01}), completing the proof of Theorem~\ref{thm1}.

\bibliographystyle{chicago}
\bibliography{biblist}

\end{document}